\documentclass[11pt]{article}
\usepackage{amsmath,amssymb,amsthm}

\usepackage[left=2.2cm,top=2.5cm,right=2.2cm,bottom= 4cm]{geometry}

\newtheorem{theorem}{Theorem}[section]

\newtheorem{proposition}[theorem]{Proposition}

\newtheorem{lemma}[theorem]{Lemma}
\newtheorem{remark}[theorem]{Remark}
\newtheorem{example}[theorem]{Example}
\newtheorem{examples}[theorem]{Examples}
\newtheorem{foo}[theorem]{Remarks}

\newenvironment{Examples}{\begin{examples}\rm}{\end{examples}}

 \newcommand{\be}{\begin{equation}}
 \newcommand{\ee}{\end{equation}}
 \newcommand{\bea}{\begin{eqnarray}}
 \newcommand{\eea}{\end{eqnarray}}
 \newcommand{\beas}{\begin{eqnarray*}}
\newcommand{\eeas}{\end{eqnarray*}}

\newcommand{\brak}[1]{\ensuremath{\left( #1 \right)}}
\newcommand{\crl}[1]{\ensuremath{ \left\{ #1 \right\} }}
\newcommand{\edg}[1]{\ensuremath{ \left[ #1 \right] }}
\newcommand{\ang}[1]{\ensuremath{ \left \langle #1 \right \rangle }}

\newcommand{\p}{\mathbb{P}}
\newcommand{\q}{\mathbb{Q}}

\newcommand{\EP}{{\mathbb E}^{\p}}
\newcommand{\EQ}{{\mathbb E}^{\q}}

\newcommand{\AV}{{\rm AVaR}_{\lambda}}
\newcommand{\Ent}{{\rm Ent}_{\lambda}}

\newcommand{\BIGOP}[1]{\mathop{\mathchoice%
{\raise-0.22em\hbox{\huge $#1$}}%
{\raise-0.05em\hbox{\Large $#1$}}{\hbox{\large $#1$}}{#1}}}

\newcommand{\BIGboxplus}{\mathop{\mathchoice%
{\raise-0.35em\hbox{\huge $\boxplus$}}%
{\raise-0.15em\hbox{\Large $\boxplus$}}{\hbox{\large $\boxplus$}}{\boxplus}}}

\begin{document}

\title{Duality formulas for robust pricing and\\ hedging in discrete time\footnote{We thank Daniel Bartl,
Peter Carr, Samuel Drapeau, Marek Musiela, Jan Ob{\l}{\'o}j, Mete Soner and Nizar Touzi for 
fruitful discussions and helpful comments. The first author was partially supported by 
NSF Grant DMS-1515753 and the third one by Vienna Science and Technology Fund Grant MA 14-008.}}

\author{Patrick Cheridito\footnote{Department of Mathematics, ETH Zurich, 8092 Zurich, Switzerland.}
\and Michael Kupper\footnote{Department of Mathematics and Statistics, University of Konstanz, 
78464 Konstanz, Germany.}
\and Ludovic Tangpi\footnote{Faculty of Mathematics, University of Vienna, 1090 Vienna, Austria.}}

\date{September 2017}

\maketitle

\begin{abstract}

\noindent
In this paper we derive robust super- and subhedging dualities for contingent claims that can depend
on several underlying assets. In addition to strict super- and subhedging, we also consider 
relaxed versions which, instead of eliminating the shortfall risk completely, aim to reduce it to an 
acceptable level. This yields robust price bounds with tighter spreads. As examples
we study strict super- and subhedging with general convex transaction costs and trading constraints
as well as risk-based hedging with respect to robust versions of the average value at risk and 
entropic risk measure. Our approach is based on representation results for increasing convex functionals and 
allows for general financial market structures. As a side result it yields a robust version of 
the fundamental theorem of asset pricing.\\[2mm]
\end{abstract}

\setcounter{equation}{0}
\section{Introduction}
\label{sec:intro}

Super- and subhedging dualities lie at the heart of no-arbitrage arguments in 
quantitative finance. By relating prices to hedging, they provide bounds on 
arbitrage-free prices. But they also serve as a stepping stone to the application of
duality methods to portfolio optimization problems. In traditional financial modeling,
uncertainty is described by a single probability measure $\p$, and the 
super- and subhedging prices of a contingent claim with discounted payoff $X$ are given by 
\be \label{suph}
\phi(X) = \inf \crl{m \in \mathbb{R} : \mbox{there exists a } Y \in G \mbox{ such that }
m-X+ Y \ge 0 \mbox{ $\p$-a.s.}}
\ee
and
\be \label{subh}
-\phi(-X) = \sup \crl{m \in \mathbb{R} : \mbox{there exists a } Y \in G \mbox{ such that }
X-m + Y \ge 0 \mbox{ $\p$-a.s.}},
\ee
where $G$ is the set of all realizable discounted trading gains. Classical results assume that $X$ depends 
on a set of underlying assets which can be traded dynamically without transaction costs or constraints. 
Then $G$ is a linear space, and under a suitable no-arbitrage condition, one obtains dualities of the form
\be \label{classical}
\phi(X) = \sup_{\q \in {\cal M}^e(\p)} \EQ X \quad \mbox{and} \quad 
-\phi(-X) = \inf_{\q \in {\cal M}^e(\p)} \EQ X,
\ee
where ${\cal M}^e(\p)$ is the set of all (local) martingale measures equivalent to $\p$; see e.g.
\cite{FS} or \cite{DSch} for an overview.

In this paper we do not assume that the probabilities of all future events are known. 
So instead of starting with a predefined probability measure, we 
specify a collection of possible trajectories for a set of basic assets.
This includes a wide range of setups, from a single binomial tree model to the model-free case,
in which at any time, the prices of all assets can lie anywhere in $\mathbb{R}_+$ (or $\mathbb{R}$). 
We replace the $\p$-almost sure inequalities in \eqref{suph} and \eqref{subh} by a general set $A$ of 
acceptable discounted positions and consider super- and subhedging functionals of the form 
\be \label{rob}
\phi(X) = \inf \crl{m \in \mathbb{R} : m - X \in A- G} \quad \mbox{and} \quad -\phi(-X)
= \sup \crl{m \in \mathbb{R} : X-m \in A- G}.
\ee
If $A$ is the cone of non-negative outcomes, this describes strict super- and subhedging, which 
requires that the shortfall risk be eliminated completely. Alternatively, one can 
allow for a certain amount of risk by enlarging the set $A$. This reduces the spread between 
super- and subhedging prices. Moreover, the set of discounted trading gains $G$ does not have to be a linear space and 
can describe general market structures with transaction costs and trading constraints. 

Our main result, Theorem \ref{thm:main}, provides dual representations for the quantities in \eqref{rob} in terms 
of expected values of $X$. As a byproduct, it yields a robust fundamental 
theorem of asset pricing (FTAP), which relates two different notions of no-arbitrage to the existence of 
generalized martingale measures. In the case where the underlying assets are bounded, it holds for 
general sets $A$ and $G$ such that $G-A$ is convex. Otherwise, it needs that the set $G-A$ is large 
enough. This can be guaranteed by assuming that the financial market is sufficiently rich or, as 
shown in Proposition \ref{prop:A}, that the acceptability condition is not too strict. 
In Sections \ref{sec:ss} and \ref{sec:rm} we study different specifications of $A$ and $G$, for which the
dual representations can be computed explicitly. Section \ref{sec:ss} is devoted to the case where 
$A$ consists of all non-negative outcomes, corresponding to strict super- and subhedging.
We consider general semi-static trading strategies consisting of dynamic investments in the underlying assets
and static derivative positions. Proposition \ref{prop:trans} covers general convex transaction 
costs and constraints on the derivative holdings. Proposition \ref{prop:short} deals with dynamic 
shortselling constraints. In Section \ref{sec:rm} we relax the hedging requirement and control shortfall risk
with a family of risk measures defined via different probability measures.
This allows to introduce risk-tolerance in a setup of model-uncertainty.
Our price bounds then become robust good deal bounds.
Proposition \ref{prop:AVaR} gives an explicit duality formula in the case where shortfall risk is assessed with 
a robust average value at risk. Proposition \ref{prop:Ent} provides the same for a robust entropic risk measure. 

Our approach is based on representation results for increasing convex functionals that we develop in the 
appendix. It permits to combine robust methods with transaction costs, trading constraints, partial hedging and good deal
bounds. Robust hedging methods go back to \cite{Hob98} and were further investigated in e.g.
\cite{Cox-Wang,dav-obl-rav,obloj}. Various versions of robust FTAPs and superhedging dualities have been derived in 
\cite{Acciaio2013,ban-dol-goe,BCKT,bei-hl-pen,Bou-Nutz,Burz,BMF,bur-fri-mag,fabio15,DS_transport,
DS,FH,gal-hl-tou,HO,Riedel}. For FTAPs and superhedging dualities under transaction costs we refer to 
\cite{CJP, DS13, grs, JK95, Schachermayer04} and the references therein.
The literature on partial hedging started with the quantile-hedging approach of \cite{Foel-Leu} and 
subsequently developed more general risk-based methods; see e.g. \cite{Carr01,Foel-Leu2000,Rud07}
and the closely related literature on good deals, such as e.g. \cite{Nadal,JK,KS,Staum04}.

The rest of the paper is organized as follows: Section \ref{sec:main} introduces the notation and states 
the paper's main results. In Section \ref{sec:ss} we study strict super- and subhedging and give 
robust FTAPs with corresponding superhedging dualities under convex transaction costs and 
trading constraints. As special cases we obtain versions of Kantorovich's transport duality 
\cite{Kanto} that include martingale or supermartingale constraints. 
In Section \ref{sec:rm} we use robust risk measures to weaken the acceptability condition.
This yields robust good deal bounds lying closer together than the strict 
super- and subhedging prices of Section \ref{sec:ss}. All proofs are given in the appendix.

\setcounter{equation}{0}
\section{Main results}
\label{sec:main}

We consider a model with $J+1$ financial assets $S^0, \dots, S^J$ and finitely many trading periods $t = 0,1, \dots, T$. 
As sample space we take a non-empty subset $\Omega$ of 
$(\mathbb{R}_{++} \times \mathbb{R}^J)^T$, where $\mathbb{R}_{++}$ denotes the 
positive half-line $(0,+\infty)$. The initial prices of the 
assets are assumed to be known and given by $S^0_0 = 1$ together with numbers $S^j_0 \in \mathbb{R}$,
$j = 1, \dots, J$. Their future prices are uncertain and modeled as $S^j_t(\omega) = \omega^j_t$, 
$t =1, \dots, T$, $j=0, \dots, J$, $\omega \in \Omega$. To state and prove the results in this paper, it will be 
convenient to quote prices in units of $S^0$. Discounted like this, the asset prices become 
$\tilde{S}^0 \equiv 1$ and $\tilde{S}^j_t = S^j_t/S^0_t$, $j =1 , \dots, J$. However, derivatives on 
$S^0, \dots, S^J$ are usually specified in terms of the nominal prices $S^j_t$ and not the discounted prices $\tilde{S}^j_t$.
So the discounted payoff of a general contingent claim on $S^0, \dots, S^J$ with maturity $T$ is
given by a mapping $X : \Omega \to \mathbb{R}$.

We endow $\Omega$ with the Euclidean metric and denote the space of all Borel measurable 
functions $X : \Omega \to \mathbb{R}$ by $B$. All discounted trading gains that can 
be realized by investing in the financial market from time $0$ until $T$ are given by a subset 
$G \subseteq B$ containing $0$. $G$ might contain discounted payoffs of derivatives maturing 
before time $T$. Then the proceeds are invested in $S^0$ and held until time $T$.
All acceptable discounted time-$T$ positions are modeled with a subset $A \subseteq B$ containing 
the positive cone $B^+ := \crl{X \in B : X \ge 0}$ such that $A + B^+ \subseteq A$ and $A-G$ is convex.
The corresponding superhedging functional is
$$
\phi(X) := \inf \crl{m \in \mathbb{R} : m-X \in A-G}, \quad \mbox{where } \inf \emptyset := + \infty.
$$
It determines for every liability with discounted time-$T$ payoff $X \in B$, the minimal initial capital needed
such that $m -X$ can be transformed into an acceptable position by investing in the financial market. 
The case $A = B^+$ corresponds to strict superhedging, which requires that 
a contingent claim be superreplicated in every possible scenario $\omega \in \Omega$. A larger acceptance set $A$
relaxes the superhedging requirement and lowers the hedging costs. The subhedging 
functional induced by $G$ and $A$ is given by
$$
-\phi(-X) = \sup \crl{m \in \mathbb{R} : X - m \in A -G}, \quad \mbox{where } \sup \emptyset := - \infty.
$$

We assume that there exists a continuous function $Z : \Omega \to [1, +\infty)$ with compact sublevel 
sets\footnote{It follows from this assumption that $\Omega$ is $\sigma$-compact. In particular, it is a 
Borel measurable subset of $(\mathbb{R}_{++} \times \mathbb{R}^J)^T$. On the other hand, if $\Omega$ is a 
non-empty subset of $(\mathbb{R}_{++} \times \mathbb{R}^J)^T$ that is closed in $\mathbb{R}^{(J+1)T}$, 
any continuous function $Z \colon \Omega \to [1, + \infty)$ with bounded sublevel sets  has compact sublevel sets.} 
$\crl{\omega \in \Omega: Z(\omega) \le z}$ for all $z \in \mathbb{R}_+$. 
We will consider hedging dualities for discounted payoff functions whose growth is controlled by $Z$.
Let $B_Z \subseteq B$ be the subspace consisting of functions $X \in B$ such that 
$X/Z$ is bounded, $U_Z$ the set of all upper semicontinuous $X \in B_Z$ 
and $C_Z$ the space of all continuous $X \in B_Z$. By ${\cal P}_Z$ we denote the set of all Borel 
probability measures on $\Omega$ satisfying the integrability condition $\EP Z < + \infty$.
Define the convex conjugate $\phi^* : {\cal P}_Z \to \mathbb{R} \cup \crl{\pm \infty}$ by
$$
\phi^*(\p) := \sup_{X \in C_Z} (\EP X - \phi(X)).
$$
Then the following holds:

\begin{theorem} \label{thm:main}
Assume 
\be \label{condcall}
\mbox{for every $n \in \mathbb{N}$, there exists a $z \in \mathbb{R}_+$ such that } \;
n(Z - z)^+ - \frac{1}{n} \in G-A.
\ee
Then the following three conditions are equivalent:
\begin{itemize}
\item[{\rm (i)}] 
there exist no $X \in G-A$ and $\varepsilon \in \mathbb{R}_{++}$ such that $X \ge \varepsilon$
\item[{\rm (ii)}]
there exists a probability measure $\p \in {\cal P}_Z$ such that $\EP X \le 0$ for all $X \in C_Z \cap (G-A)$
\item[{\rm (iii)}] 
$\phi$ is real-valued on $B_Z$ with $\phi(0) = 0$ and
$\phi(X) = \max_{\p \in {\cal P}_Z} (\EP X - \phi^*(\p))$ for all $X \in C_Z$.
\end{itemize}
If in addition to \eqref{condcall}, one has
\be \label{condusc}
\phi(X) = \inf_{Y \in C_Z, Y \ge X} \phi(Y) \quad \mbox{for all } X \in U_Z,
\ee
then {\rm (i)--(iii)} are also equivalent to each of the following three:
\begin{itemize}
\item[{\rm (iv)}] 
there exists no $X \in G-A$ such that $X(\omega) > 0$ for all $\omega \in \Omega$
\item[{\rm (v)}]
there exists a probability measure $\p \in {\cal P}_Z$ such that $\EP X \le 0$ for all $X \in U_Z \cap (G-A)$
\item[{\rm (vi)}] 
$\phi$ is real-valued on $B_Z$ with $\phi(0) = 0$ and
$\phi(X) = \max_{\p \in {\cal P}_Z} (\EP X - \phi^*(\p))$ for all $X \in U_Z$.
\end{itemize}
\end{theorem}

\eqref{condcall} and \eqref{condusc} are both conditions on the set $G-A$ (the latter since $\phi$ is defined by $G-A$).
\eqref{condcall} is trivially satisfied if $\Omega$ is compact since in this case, $(Z-z)^+=0$ for $z \in \mathbb{R}_+$ large enough.
On the other hand, if $\Omega$ is not compact, \eqref{condcall} holds if, for instance, for every $n \in \mathbb{N}$, 
there exists a $z \in \mathbb{R}_+$ such that there is an investment opportunity yielding 
a discounted outcome of at least $n(Z-z)^+$ at 
an initial cost of no more than $1/n$, or alternatively, if $1/n - n(Z-z)^+$ is considered to be an 
acceptable discounted position. Proposition \ref{prop:A} below provides a class of acceptance 
sets such that \eqref{condcall} holds without additional assumptions on $G$, and Proposition 
\ref{prop:dualcond} gives an equivalent condition for \eqref{condusc}.
In all our examples in Sections \ref{sec:ss} and \ref{sec:rm} below, both conditions, \eqref{condcall} and 
\eqref{condusc}, are satisfied.

(i) and (iv) are no-arbitrage conditions, or in the case where the acceptance set $A$ is larger than 
$B^+$, so called no-good deal conditions. (i) means that there exists no trading 
strategy starting with zero initial capital generating an outcome that exceeds an
acceptable position by a positive fraction of $S^0_T$, or equivalently, no trading 
strategy turning a negative initial wealth into an acceptable position. The same condition was used by 
\cite{Garman85} and \cite{Lewis10} in the classical framework. (iv) is slightly stronger. In the case where 
$A$ equals $B^+$, it corresponds to absence of model-independent arbitrage 
as introduced by \cite{Dav-Hob07} and used e.g., in \cite{Acciaio2013}. We point out that for $A = B^+$, 
(i) and (iv) are both weaker than the traditional no-arbitrage condition, which requires that there exist
no trading strategies generating a non-negative profit that is positive on a non-negligible part 
of the sample space (see e.g. \cite{HK, HP, FS, DS}).

(ii) and (v) generalize the concept of a martingale measure. For instance, if $A = B^+$ and the underlying assets 
are liquidly traded, they consist of proper martingale measures. But in the presence of proportional 
transaction costs, they become $\varepsilon$-approximate martingale measures, and 
under short-selling constraints, supermartingale measures (see the examples in Section \ref{sec:ss}).

(iii) and (vi) yield dual representations for the superhedging functional $\phi$. The max means that 
the right sides of (iii) and (vi) are suprema which are attained.
(iii) and (vi) directly translate into dual representations for the subhedging functional $-\phi(-X)$. If condition (iii) holds, one has
$$
-\phi(-X) = \min_{\p \in {\cal P}_Z} (\EP X + \phi^*(\p)) \quad \mbox{for all } X \in C_Z,
$$
and the representation extends to all $X \in U_Z$ if (vi) is satisfied. Moreover, note that
as soon as $\phi$ is real-valued on $B_Z$ with $\phi(0) = 0$, the same is true 
for the subhedging functional, and one obtains by convexity,
$\phi(X) + \phi(-X) \ge 2 \phi(0) = 0$, yielding the ordering
$$\phi(X) \ge - \phi(-X) \quad \mbox{for all } X \in B_Z.
$$

A wide class of acceptance sets can be written as
\be \label{acc}
A = \crl{X \in B_Z : \EP X + \alpha(\p) \ge 0 \mbox{ for all } \p \in {\cal P}_Z} + B^+
\ee
for a suitable mapping $\alpha : {\cal P}_Z \to \mathbb{R}_+ \cup \crl{+ \infty}$. In the 
extreme case $\alpha \equiv 0$, $A$ is the positive cone $B^+$. On the other hand, it can be shown 
that if $\alpha$ grows fast enough, assumption \eqref{condcall} of Theorem \ref{thm:main} is 
automatically satisfied:

\begin{proposition} \label{prop:A}
Condition \eqref{condcall} holds if $A$ is given by \eqref{acc} for a mapping 
$\alpha : {\cal P}_Z \to \mathbb{R}_+ \cup \crl{+\infty}$ satisfying
\begin{enumerate}
\item[{\rm (A1)}] $\inf_{\p \in {\cal P}_Z} \alpha(\p) =0$ and
\item[{\rm (A2)}] $\alpha(\p) \ge \EP \beta(Z)$ for all $\p \in {\cal P}_Z$, where
$\beta : [1,+\infty) \to \mathbb{R}$ is an increasing\footnote{We call a function $f$ from a subset 
$I \subseteq \mathbb{R}$ to $\mathbb{R}$ increasing if $f(x) \ge f(y)$ for $x \ge y$.} function with the property
$\lim_{x \to + \infty} \beta(x)/x = + \infty$.
\end{enumerate}
\end{proposition}

The following result gives a dual condition for assumption \eqref{condusc} which will be useful in 
Sections \ref{sec:ss} and \ref{sec:rm}.

\begin{proposition} \label{prop:dualcond}
Assume \eqref{condcall} holds. Then
$$
\phi^*(\p) = \sup_{X \in C_Z \cap (G-A)} \EP X \le  \sup_{X \in U_Z} (\EP X - \phi(X))
= \sup_{X \in U_Z \cap (G-A)} \EP X \quad \mbox{for all } \p \in {\cal P}_Z,
$$
and the inequality is an equality if and only if $\phi$ satisfies \eqref{condusc}.
\end{proposition}

\setcounter{equation}{0}
\section{Strict superhedging}
\label{sec:ss}

In this section we concentrate on the case where the acceptance set $A$ is given by the positive cone $B^+$. 
$\Omega$ is assumed to be a non-empty closed subset of $\prod_{t=1}^T ([a_t,b_t] \times \mathbb{R}^{J}_+)$
for numbers $0 < a_t \le b_t$,
and the price processes $S^0, \dots, S^J$ are given by $S^0_0 = 1$, $S^j_0 \in \mathbb{R}_+$, $j =1, \dots, J$
and $S^j_t(\omega) = \omega^j_t$, $t \ge 1$, $\omega \in \Omega$. 
They generate the filtration ${\cal F}_t = \sigma(S^j_s : j =0, \dots, J, \,s \le t)$, $t = 0, \dots, T$.
As growth function we choose $Z = 1 + \sum_{j,t \ge 1} (\tilde{S}^j_t)^p$ for a constant $p \ge 1$.
It clearly is continuous, and the sublevel sets $\crl{\omega \in \Omega : Z(\omega) \le z}$ are compact
for all $z \in \mathbb{R}_+$. Moreover, $\tilde{S}^j_t$ belongs to $C_Z$ for all $j$ and $t$.
For any set $G \subseteq B$ of discounted trading gains containing 
$0$ such that $G - B^+$ is convex, condition \eqref{condcall} is equivalent to 
\be \label{sscondition}
\begin{aligned}
&\mbox{for every $n\in\mathbb{N}$, $j=1,\dots,J$, and $t=1,\dots, T$, there exist}\\[-1mm]
&\mbox{$X \in G$ and $K \in \mathbb{R}_+$ such that $X \ge n((\tilde{S}^j_t)^p-K)^+ - 1/n$.}
\end{aligned}
\ee
For instance, in the case $p=1$, condition \eqref{sscondition} holds if the market offers call options 
on all assets $S^1, \dots, S^J$ with every maturity $t = 1, \dots, T$ at arbitrarily small prices.

Let us define
$$
\phi^*_G(\p) := \sup_{X \in G} \EP X, \quad \p \in {\cal P}_Z,
$$ with the understanding that 
$$
\EP X := \left\{
\begin{array}{ll}  \EP X^+ - \EP X^- & \mbox{ if } \EP X^- < + \infty\\
- \infty & \mbox{ otherwise,}
\end{array}
\right.
$$
and introduce the corresponding set of generalized martingale measures
$$
{\cal M} := \left\{\p\in {\cal P}_Z: \phi^*_G(\p)=0\right\}=\left\{\p\in {\cal P}_Z:\EP X \le 0
\mbox{ for all } X \in G \right\}.
$$

\subsection{Semi-static hedging with convex transaction costs and constraints}
\label{subsec:semis}

Let us first assume that the assets $S^0, \dots, S^J$ can be traded dynamically, and in addition,
it is possible to form a static portfolio of derivatives depending on $S^0, \dots, S^J$. We describe the dynamic 
part of the trading strategy by a $J$-dimensional predictable process $(\vartheta_t)_{t=1}^T$
modeling the holdings of the assets $S^1, \dots, S^J$ over time and suppose that
buying or selling shares of $S^j$ at time $t$ incurs transaction costs of the form 
$g^j_t(\omega, \Delta\vartheta^j_{t+1}(\omega) S^j_t(\omega))$ for a continuous 
function $g^j_t : \Omega \times \mathbb{R} \to \mathbb{R}$
such that $g^j_t(\omega, x)$ is ${\cal F}_t$-measurable in $\omega$ and convex in 
$x$ with $g^j_t(\omega,0) =0$. In the case where $\Omega$ is an unbounded subset of $\mathbb{R}^{(J+1)T}$,
we also assume $\sup_{x \in E} |g^j_t(\omega,x)|/Z(\omega)$ to be bounded in $\omega$
for every bounded subset $E \subseteq \mathbb{R}$. The static portfolio can be formed
by investing in a given set of derivatives with discounted payoffs $H_i \in C_Z$, $i \in I$.
We make no assumptions on the index set $I$. In particular, $(H_i)_{i \in I}$ can be
an infinite collection. However, it is only possible to 
invest in finitely many of them. More precisely, we denote by $\mathbb{R}^I_0$ the set of vectors in 
$\mathbb{R}^I$ with at most finitely many components different from $0$ and suppose that the static part 
of the strategy $\theta$ is constrained to lie in a given convex subset $\Theta \subseteq \mathbb{R}^I_0$ containing $0$. 
Discounted transaction costs in the derivatives market are given by a convex mapping $h : \Theta \to C_Z$ satisfying $h(0) = 0$.
This means that they can depend on the underlying uncertainty.
As usual, we suppress the $\omega$-dependence of $g^j_t$ in the notation. Then 
the resulting set of discounted trading gains $G$ consists of outcomes of the form 
$$
\sum_{t=1}^T\sum_{j=1}^J \brak{\vartheta_t^j \Delta \tilde{S}_t^j - 
\frac{g^j_{t-1}(\Delta\vartheta^j_{t} S^j_{t-1})}{S^0_{t-1}}} + \sum_{i \in I} \theta_i H_i -  h(\theta),
$$
where $\Delta \tilde{S}^j_t = \tilde{S}^j_t - \tilde{S}^j_{t-1}$ and $\Delta \vartheta^j_t = \vartheta^j_t - \vartheta^j_{t-1}$ 
with $\vartheta^j_0=0$.

For this specification of our general model, the following can be deduced from Theorem \ref{thm:main}
and Proposition \ref{prop:dualcond}:

\begin{proposition} \label{prop:trans}
If $G$ satisfies condition \eqref{sscondition}, the following four conditions are equivalent:
\begin{itemize}
\item[{\rm (i)}] 
there exist no $X\in G$ and $\varepsilon \in \mathbb{R}_{++}$ such that $X\ge\varepsilon$
\item[{\rm (ii)}] 
there exists no $X \in G$ such that $X(\omega) > 0$ for all $\omega \in \Omega$
\item[{\rm (iii)}] 
${\cal M}\neq\emptyset$
\item[{\rm (iv)}] $\phi$ is real-valued on $B_Z$ with $\phi(0) = 0$ and
$\phi(X) = \max_{\p \in {\cal P}_Z} (\EP X - \phi^*_G(\p))$ for all $X \in U_Z$.
\end{itemize} 
Moreover,
\begin{equation} \label{penaltyG}
\phi^*_G(\p) =
\left\{\begin{array}{ll}
\sum_{t=0}^{T-1}\sum_{j=1}^J \EP \edg{\frac{1}{S^0_t} 
g^{j*}_t \left(\frac{\EP[\tilde{S}^j_T\mid{\cal F}_{t}]- \tilde{S}^j_{t}}{\tilde{S}^j_{t}} \right)1_{\crl{\tilde{S}^j_{t}>0}}}
+ h^*(\p) & \mbox{ if } \p \in {\cal R}\\
+ \infty & \mbox{ if } \p \notin {\cal R},
\end{array} \right.
\end{equation}
for $$
g^{j*}_t(y) := \sup_{x \in \mathbb{R}}(xy - g^j_t(x)), \quad
h^*(\p) := \sup_{\theta \in\Theta} \EP \brak{\sum_{i \in I}\theta_i  H_i-h(\theta)},
$$
and
$$
{\cal R} := \crl{\p \in {\cal P}_Z :
\p[S^j_t =0 \mbox{ and } S^j_T > 0]=0 \mbox{ for all } j =1, \dots, J \mbox{ and } t \le T-1}.
$$
\end{proposition}

Proposition \ref{prop:trans} extends Theorem 2.2 of \cite{ban-dol-goe}, which provides a 
duality result for general convex transaction costs in a model with a compact sample space 
in which there are no derivatives available for hedging.

\subsubsection{Proportional transaction costs}

As a special case, let us consider proportional transaction costs together with the constraint 
that some of the derivatives in the market might only be available to buy or sell. More precisely,
dynamic transaction costs are given by functions of the form $g^j_t(x) = \varepsilon^j_t |x|$ for  
${\cal F}_t$-measurable random coefficients $\varepsilon^j_t \in C^+_Z$. The static part of the 
hedging portfolio is given by a vector $\theta \in \mathbb{R}^I_0$ with associated cost
$h(\theta) = \sum_{i \in I} h^+_i \theta^+_i - h^-_i \theta^-_i$ for bid and ask prices $h^-_i \in \mathbb{R} \cup \crl{-\infty}$ and 
$h^+_i \in \mathbb{R} \cup \crl{+\infty}$, where $h^+_i = + \infty$ means that $H_i$ cannot be bought
and $h^-_i = - \infty$ that it cannot be sold. The corresponding discounted trading outcomes are of the form
$$
\sum_{t=1}^T\sum_{j=1}^J \brak{\vartheta_t^j \Delta \tilde{S}_t^j - \varepsilon^j_{t-1} |\Delta\vartheta^j_{t} \tilde{S}^j_{t-1}|}
+ \sum_{i \in I} \brak{\theta_i H_i -  \theta^+_i h^+_i + \theta^-_i h^-_i},
$$
and
$$
g^{j*}_t(y) = \left\{
\begin{array}{ll}
0 & \mbox{ if } |y| \le \varepsilon^j_t\\
+ \infty & \mbox{ otherwise,}
\end{array} \right. \qquad
h^*(\p) = \left\{
\begin{array}{ll}
0 & \mbox{ if } h^-_i \le \EP H_i \le h^+_i \mbox{ for all } i \in I\\
+ \infty & \mbox{ otherwise.}
\end{array} \right.
$$
By Proposition \ref{prop:trans}, one has 
$$
\phi^*_G(\p) = \left\{
\begin{array}{ll}
0 & \mbox{ if } \p \in {\cal M}\\
+ \infty & \mbox{ otherwise,}
\end{array} \right.
$$
where ${\cal M}$ is the set of all probability measures $\p \in {\cal P}_Z$ satisfying the conditions
\begin{itemize}
\item[a)]
$(1-\varepsilon^j_t) \tilde{S}^j_t \le \EP[\tilde{S}^j_T \mid {\cal F}_t] \le (1+\varepsilon^j_t) \tilde{S}^j_t$
for all $j = 1, \dots, J$ and $t = 0, \dots, T-1$
\item[b)]
$h^-_i \le \EP H_i \le h^+_i \mbox{ for all } i \in I$.
\end{itemize}
So if \eqref{sscondition} and (i) of Proposition \ref{prop:trans} hold, one obtains the duality
\be \label{propduality}
\phi(X) = \max_{\p \in {\cal M}} \EP X \quad \mbox{for all } X \in U_Z.
\ee
If $\varepsilon^j_t \equiv 0$, dynamic trading is frictionless, and a) reduces to the standard 
martingale condition. In this case, \eqref{propduality} is the superhedging duality 
given in Theorem 1.4 of \cite{Acciaio2013} for $S^0 \equiv 1$. On the other hand, in the case 
where the coefficients $\varepsilon^j_t$ are constant 
and $(H_i)_{i \in I}$ consists of European options depending continuously on $\tilde{S}_T$, \eqref{propduality} becomes the 
superhedging duality shown in Theorem 2.6 of \cite{DS} for $S^0 \equiv 1$.

\subsubsection{Superlinear transaction costs}

For $\Theta = \mathbb{R}^I_0$ and transaction costs corresponding to 
$$
g^j_t(x)= \frac{\varepsilon^j_t}{p_j} |x|^{p_j}, \quad h(\theta) = \sum_{i \in I} h_i \theta_i + \frac{\delta_i}{q_i} |\theta_i|^{q_i}
$$
for positive ${\cal F}_t$-measurable $\varepsilon^j_t \in C_Z$ and constants 
$\delta_i >0$, $p_j,q_j > 1$, $h_i \in \mathbb{R}$, one obtains from Proposition \ref{prop:trans},
\[
\phi^*_G(\p)=\sum_{t=0}^{T-1} \sum_{j=1}^J \EP \edg{ \frac{(\varepsilon^j_t)^{1-p'_j}}{S^0_t p'_j}
\left| \frac{\EP[\tilde{S}^j_T \mid{\cal F}_t]-\tilde{S}^j_t}{\tilde{S}^j_t} \right|^{p'_j}}
+ \sum_{i \in I} \frac{\delta_i^{1-q_i'}}{q_i'} \left|\EP H_i-h_i\right|^{q_i'},
\]
where $p_j' := p_j/(p_j-1)$, $q_i' := q_i/(q_i-1)$, $0/0 := 0$ and $x/0 := + \infty$ for $x > 0$. 
Moreover, if \eqref{sscondition} and condition (i) of Proposition \ref{prop:trans} hold, one has
$$
\phi(X) = \max_{\p \in {\cal P}_Z} (\EP X - \phi^*_G(\p))\quad \mbox{for all }X \in U_Z.
$$

\subsubsection{European call options and constraints on the marginal distributions}

Let $g^j_t : \Omega \times \mathbb{R} \to \mathbb{R}$ be as in the beginning of 
Subsection \ref{subsec:semis} above and assume the family $(H_i)_{i \in I}$ consists of all 
discounted European call option payoffs
$$
(\tilde{S}^j_t-K)^+, \quad j = 1, \dots J, \; t = 1, \dots, T, \; K \in \mathbb{R}_+.
$$
Moreover suppose that arbitrary quantities of options with discounted payoffs $(\tilde{S}^j_t-K)^+$ can be bought or sold at prices 
$p^{j,+}_{t,K}$ and $p^{j,-}_{t,K}$, respectively. If $\lim_{K \to + \infty} p^{j,+}_{t,K} = 0$ for all 
$j$ and $t$, condition \eqref{sscondition} holds. So if in addition, (i) of Proposition \ref{prop:trans} is satisfied, 
one obtains
\be \label{gmt}
\phi(X) = \max_{\p} \brak{\EP X - \sum_{t=0}^{T-1} \sum_{j=1}^J 
\EP \edg{ \frac{1}{S^0_t} g^{j*}_{t} \brak{\frac{\EP[\tilde{S}^j_T\mid{\cal F}_{t}]- \tilde{S}^j_{t}}{\tilde{S}^j_{t}}}
1_{\crl{\tilde{S}^j_{t}>0}}}}, \quad \mbox{for all } X \in U_Z,
\ee
where the maximum is over all $\p \in {\cal P}_Z$ such that 
\begin{itemize}
\item[a)] $\p[S^j_{t} = 0 \mbox{ and } S^j_T > 0] = 0$ for all $j = 1, \dots, J$ and $t = 0, \dots, T-1$
\item[b)] $p^{j,-}_{t,K} \le \EP (\tilde{S}^j_t - K)^+ \le p^{j,+}_{t,K}$ for all $j=1, \dots, J,$ $t=1, \dots T$ and $K \in \mathbb{R}_+$.
\end{itemize}
$\EP (\tilde{S}^j_t - K)^+$ can be written as $\int_K^{+\infty} \p[\tilde{S}^j_t > x] dx$.
So condition b) puts constraints on the distributions of $\tilde{S}^j_t$ under $\p$, and 
in the limiting case $p^{j,+}_{t,K} = p^{j,-}_{t,K}$, it fully determines the distributions of  
$\tilde{S}^j_t$ under $\p$. In particular, if $g^j_t \equiv 0$ and $p^{j,+}_{t,K} = p^{j,-}_{t,K} =p^j_{t,K} \in \mathbb{R}_+$ for all 
$j,t,K$, it follows from \eqref{sscondition} and (i) of Proposition \ref{prop:trans} that there exist unique 
marginal distributions $\nu^j_t$ on $\mathbb{R}_+$ specified by $\int_K^{+\infty} \nu^j_t(x,+\infty) dx = p^j_{t,K}$,
$K \in \mathbb{R}_+$, such that
\be \label{mt}
\phi(X) = \max_{\p \in {\cal M}} \EP X \quad \mbox{for all } X \in U_Z,
\ee 
where ${\cal M}$ consists of all $\p \in {\cal P}_Z$ satisfying
\begin{itemize}
\item[a)] $\tilde{S}^1, \dots, \tilde{S}^J$ are martingales under $\p$
\item[b)] the distribution of $\tilde{S}^j_t$ under $\p$ is $\nu^j_t$ for all $j =1, \dots, J$ and $t = 1, \dots, T$.
\end{itemize}
\eqref{mt} is a variant of Kantorovich's transport duality \cite{Kanto} and has recently been studied in 
different setups under the name martingale transport duality; see e.g., \cite{BCKT,bei-hl-pen,gal-hl-tou}.

\subsection{Semi-static hedging with short-selling constraints}

Now assume that dynamic trading does not incur transaction costs, but only non-negative quantities of the 
assets $S^1, \dots, S^J$ can be held. As above, one can invest statically in derivatives with discounted payoffs
$H_i \in C_Z$, $i \in I$, according to a strategy $\theta$ lying in a convex subset 
$\Theta \subseteq \mathbb{R}^I_0$ containing $0$. Let $h : \Theta \to C_Z$ 
be a convex mapping with $h(0) = 0$ and suppose that $G$ consists of discounted outcomes of the form
\[
\sum_{t=1}^T \sum_{j=1}^J \vartheta^j_t \Delta \tilde{S}^j_t + \sum_{i \in I} \theta_i H_i - h(\theta),
\]
where $(\vartheta_t)_{t=1}^T$ is a non-negative $J$-dimensional predictable strategy and $\theta \in \Theta$.
Then the following variant of Proposition \ref{prop:trans} holds:

\begin{proposition} \label{prop:short}
If \eqref{sscondition} is satisfied, the conditions {\rm (i)--(iv)} of Proposition \ref{prop:trans} 
are equivalent, where 
$$
\phi^*_G(\p) =
\left\{\begin{array}{ll}
\sup_{\theta \in\Theta} \EP (\sum_{i \in I} \theta_i H_i-h(\theta))
& \mbox{ if } \tilde{S}^1, \dots, \tilde{S}^J \mbox{ are supermartingales under } \p\\
+\infty & \mbox{ otherwise,}
\end{array} \right.
$$
and ${\cal M} = \crl{\p \in {\cal P}_Z : \phi^*_G(\p) =0}$.
\end{proposition}

\subsubsection{Dynamic and static short-selling constraints}

If there are short-selling constraints on the dynamic as well as static part of the trading strategy,
that is, $\Theta = \mathbb{R}^I_0 \cap \mathbb{R}^I_+$, and $h(\theta) = \sum_{i \in I} h_i \theta_i$ for 
prices $h_i \in \mathbb{R}$, it follows by Proposition \ref{prop:short} from \eqref{sscondition} and 
condition (i) of Proposition \ref{prop:trans} that $\phi(X) = \max_{\p \in {\cal M}} \EP X$ for all $X \in U_Z$,
where ${\cal M}$ is the set of all measures $\p \in {\cal P}_Z$ such that 
\begin{itemize}
\item[{\rm a)}] $\tilde{S}^1, \dots, \tilde{S}^J$ are supermartingales under $\p$ 
\item[{\rm b)}] $\EP H_i \le h_i$ for all $i \in I$.
\end{itemize}

\subsubsection{Supermartingale transport duality}

Suppose now that only the dynamic part of the trading strategy is subject to short-selling constraints and 
$(H_i)_{i \in I}$ consists of all discounted call options $(\tilde{S}^j_t-K)^+$, $j=1,\dots, J$, $t = 1, \dots, T$, $K \in \mathbb{R}_+$. 
If arbitrary quantities of $(\tilde{S}^j_t - K)^+$ can be bought and sold at prices $p^{j,+}_{t,K} \le p^{j,-}_{t,K}$, respectively,
condition \eqref{sscondition} is satisfied provided that $\lim_{K \to + \infty} p^{j,+}_{t,K} = 0$ for all $j$ and $t$. 
So, if in addition, (i) of Proposition \ref{prop:trans} holds, one obtains from 
Proposition \ref{prop:short} that
\be \label{smt}
\phi(X) = \max_{\p \in {\cal M}} \EP X \quad \mbox{for all } X \in U_Z
\ee
where ${\cal M}$ consists of the measures $\p \in {\cal P}_Z$ satisfying
\begin{itemize}
\item[{\rm a)}] $\tilde{S}^1, \dots, \tilde{S}^J$ are supermartingales under $\p$ 
\item[{\rm b)}] $p^{j,-}_{t,K} \le \EP(\tilde{S}^j_t-K)^+= \int_K^{+\infty} \p[\tilde{S}^j_t > x] dx \le p^{j,+}_{t,K}$ 
for all $j =1, \dots, J$, $t = 1, \dots, T$ and $K \in \mathbb{R}_+$.
\end{itemize}
In the special case $p^{j,+}_{t,K} = p^{j,-}_{t,K} = p^j_{t,K} \in \mathbb{R}_+$, condition b) is satisfied if and only if 
for all $j$ and $t$, the distribution of $\tilde{S}^j_t$ under $\p$ is equal to a measure $\nu^j_t$ satisfying 
$\int_K^{+\infty} \nu^j_t(x, +\infty) dx = p^j_{t,K}$ for all $K \in \mathbb{R}_+$. In this case,
\eqref{smt} becomes a supermartingale version of Kantorovich's transport duality \cite{Kanto}.

\setcounter{equation}{0}
\section{Superhedging with respect to risk measures}
\label{sec:rm}

In this section we relax the strict superhedging requirement and consider sets of acceptable discounted 
outcomes of the form 
\be \label{accset}
A = \bigcap_{\q \in {\cal Q}} \crl{X \in B_Z : \rho_{\q}(X) \le 0} + B^+,
\ee
where $\mathcal{Q} \subseteq \mathcal{P}_Z$ is a non-empty set of probability measures, and for 
every $\q \in {\cal Q}$, $\rho_{\q} : B_Z \to \mathbb{R}$ is a convex risk measure. 
More specifically, we concentrate on transformed loss risk measures:
$$
\rho_{\q}(X) = \min_{s \in \mathbb{R}} \brak{\EQ l_{\q}(s-X) -s}
$$
for loss functions $l_{\q} : \mathbb{R} \to \mathbb{R}$. Up to a minus sign they coincide with the 
optimized certainty equivalents of Ben-Tal and Teboulle \cite{ben-tal01}. For unbounded random 
variables, they were studied in Section 5 of \cite{CL09}. We make the following assumptions:

\begin{itemize}
\item[(l1)] every $l_{\q}$ is increasing\footnote{i.e., $l_{\q}(x) \ge l_{\q}(y)$ for $x \ge y$} 
and convex with $\lim_{x \to \pm \infty} (l_{\q}(x) - x) = + \infty$
\item[(l2)] $\sup_{\q} \EQ l_{\q}(\varphi(Z)) < + \infty$ for an increasing function
$\varphi : [1,+\infty) \to \mathbb{R}$ satisfying \newline $\lim_{x \to +\infty} \varphi(x)/x = +\infty$
\item[(l3)] $l^*_{\q}(1) = 0$ for all $\q \in {\cal Q}$, where 
$l^*_{\q}(y) = \sup_{x \in \mathbb{R}} (xy - l_{\q}(x))$, $y \in \mathbb{R}.$
\end{itemize}
Then the following holds:

\begin{lemma} \label{lemma:oce}
For every $\q \in {\cal Q}$, $\rho_{\q}$ is a real-valued convex risk measure on $B_Z$ 
with $\rho_{\q}(0) =0$ and dual representation
\be \label{OrDual}
\rho_{\q}(X) = \max_{\p \in {\cal P}_Z, \, \p \ll \q} \brak{\EP[-X] - \EQ \edg{l^*_{\q}\brak{\frac{d\p}{d\q}}}}.
\ee
Moreover, the acceptance set $A$ given in \eqref{accset}
is of the form \eqref{acc} for a mapping $\alpha : {\cal P}_Z \to \mathbb{R}_+ \cup \crl{+\infty}$ 
satisfying {\rm (A1)} and {\rm (A2)}.
\end{lemma}

\subsection{Robust average value at risk}
\label{ss:ravar}

Now assume, as in Section \ref{sec:ss}, that $\Omega$ is a closed subset of 
$\prod_{t=1}^T ([a_t,b_t] \times \mathbb{R}^{J}_+)$ for numbers $0 < a_t \le b_t$, and consider the filtration 
${\cal F}_t:=\sigma(S^j_s : j = 0, \dots, J, \, s \le t)$ generated by $S^j_t(\omega) = \omega^j_t$. 
We choose a continuous function $Z: \Omega \to \mathbb{R}$ such that $Z \ge 1 + \sum_{j,t \ge 1} \tilde{S}^j_t$. 
Then the sublevel sets $\crl{\omega \in \Omega : Z(\omega) \le z}$, $z \in \mathbb{R}_+$, are compact, and 
all $\tilde{S}^j_t$ belong to $C_Z$.

For a fixed level $0 < \lambda \le 1$ and $\q$ from a given non-empty subset ${\cal Q} \subseteq {\cal P}_Z$,
consider the average value at risk
\[
\AV^{\mathbb Q}(X):=\frac{1}{\lambda}\int_0^{\lambda} \mbox{VaR}_u^{\mathbb Q}(X) du, \quad X \in B_Z.
\]
It is well-known (see e.g. \cite{FS}) that it can be written as
$$
\AV^{\q}(X) = \min_{s \in \mathbb{R}} \brak{\frac{\EQ(s-X)^+}{\lambda} - s}, \quad X \in B_Z,
$$
and has a dual representation of the form 
$$
\AV^{\q}(X) = \max_{\p \ll \q, \, d\p/d\q \le 1/\lambda} \EP[-X], \quad X \in B_Z.
$$
The corresponding robust acceptance set is 
$$
A = \bigcap_{\q \in {\cal Q}} \crl{X \in B_Z : \AV^{\q}(X) \le 0} + B^+ 
$$

Let us assume the assets $S^0, \dots, S^J$ can be traded dynamically subject to proportional transaction costs
if positions in $S^1, \dots, S^J$ are rebalanced given by $\varepsilon_1, \dots, \varepsilon_J \ge 0$.
Moreover, there exists a family of derivatives with discounted payoffs $(H_i)_{i \in I} \subseteq C_Z$ that can be traded statically
with bid and ask prices $h^-_i, h^+_i \in \mathbb{R}$. The resulting set of discounted trading gains $G$ consists of outcomes 
of the form 
\be \label{GAV}
\sum_{t=1}^T \sum_{j=1}^J (\vartheta^j_t \Delta \tilde{S}^j_t - \varepsilon_j |\Delta \vartheta^j_{t}\tilde{S}^j_{t-1}|) + 
\sum_{i \in I} (\theta_i H_i - \theta^+_i h^+_i + \theta^-_i h^-_i),
\ee
where $(\vartheta_t)$ is a $J$-dimensional $({\cal F}_t)$-predictable strategy
and $\theta \in \mathbb{R}^I_0$. Under these conditions, one has the following:

\begin{proposition} \label{prop:AVaR}
If ${\cal Q}$ is convex, $\sigma({\cal P}_Z,C_Z)$-closed and satisfies
\be \label{phicond}
\sup_{\q \in {\cal Q}} \EQ \varphi(Z) < + \infty \mbox{ for 
an increasing function }\varphi : [1,+\infty) \to \mathbb{R} \mbox{ such that } \lim_{x \to +\infty} 
\frac{\varphi(x)}{x} = + \infty,
\ee
the following are equivalent:
\begin{itemize}
\item[{\rm (i)}] there exist no $X \in G-A$ and $\varepsilon \in \mathbb{R}_{++}$ such that $X \ge \varepsilon$
\item[{\rm (ii)}] there exists no $X \in G-A$ such that $X(\omega) > 0$ for all $\omega \in \Omega$
\item[{\rm (iii)}] ${\cal M} \not= \emptyset$
\item[{\rm (iv)}] $\phi$ is real-valued on $B_Z$ and 
$\phi(X) = \max_{\p \in {\cal M}} \EP X$ for all $X \in U_Z$,
\end{itemize}
where ${\cal M}$ is the set of all probability measures $\p \in {\cal P}_Z$ satisfying
\begin{itemize}
\item[{\rm a)}] $(1-\varepsilon_j)\tilde{S}^j_{t} \le \EP[\tilde{S}^j_T \mid {\cal F}_{t}] \le (1+\varepsilon_j) \tilde{S}^j_{t}$ for all
$j=1, \dots, J$ and $t =0, \dots, T-1$
\item[{\rm b)}] $h^-_i \le \EP H_i \le h^+_i$ for all $i \in I$
\item[{\rm c)}] $d \p/d \q \le 1/\lambda$ for some $\q \in {\cal Q}$.
\end{itemize}
\end{proposition}

\begin{Examples} \label{ExQ}
If $Z = 1 + \sum_{t,j \ge 1} \tilde{S}^j_t$, the integrability condition \eqref{phicond} is satisfied 
by the following four families of probability measures:\\
{\bf 1.} All $\q \in {\cal P}_Z$ such that $c^j_t \le \EQ (\tilde{S}^j_t)^2 \le C^j_t$ for given constants $0 \le c^j_t \le C^j_t$.\\
{\bf 2.} All $\q \in {\cal P}_Z$ such that $c^j_t \le \EQ [(\tilde{S}^j_t/\tilde{S}^j_{t-1}-1)^2 \mid {\cal F}_{t-1}] \le C^j_t$ for 
given constants $0 \le c^j_t \le C^j_t$.\\
{\bf 3.} All $\q \in {\cal P}_Z$ under which $Y^j_t = \log(\tilde{S}^j_t/\tilde{S}^j_{t-1})$, $j=1, \dots, J$, $t=1, \dots, T$, 
forms a Gaussian family with mean vector $(\EQ Y^j_t)$ in a bounded set $M \subseteq  \mathbb{R}^{JT}$ and covariance 
matrix $\mbox{Cov}^{\q}(Y^j_t,Y^k_s)$ in a bounded set
$\Sigma \subseteq \mathbb{R}^{JT \times JT}$.\\
{\bf 4.} The $\sigma({\cal P}_Z, C_Z)$-closed convex hull of any set ${\cal Q} \subseteq {\cal P}_Z$
satisfying \eqref{phicond}.\\[2mm]
It can easily be checked that the first two families are convex and $\sigma({\cal P}_Z,C_Z)$-closed.
But the third one is in general not convex. So to satisfy the assumptions of
Proposition \ref{prop:AVaR}, one has to pass to the  $\sigma({\cal P}_Z,C_Z)$-closed convex hull.
\end{Examples}

\subsection{Robust entropic risk measure}

As in Subsection \ref{ss:ravar}, suppose that $\Omega$ is a closed subset 
$\prod_{t=1}^T ([a_t,b_t] \times \mathbb{R}^{J}_+)^T$ for numbers $0 < a_t \le b_t$, consider the filtration
$({\cal F}_t)$ generated by $(S^j_t)$, $j = 0, \dots, J$, and let $Z : \Omega \to \mathbb{R}$ 
be a continuous function such that $Z \ge 1+ \sum_{j,t \ge 1} \tilde{S}^j_t$.

For a fixed risk aversion parameter $\lambda > 0$ and $\q$ in a given non-empty set 
${\cal Q} \subseteq {\cal P}_Z$, consider the entropic risk measure
$$
\Ent^{\q}(X) = \frac{1}{\lambda} \log \EQ \exp(- \lambda X), \quad X \in B_Z.
$$
It admits the alternative representations
$$
\Ent^{\q}(X) = \min_{s \in \mathbb{R}} \brak{\frac{\exp(\lambda s - 1 - \lambda X)}{\lambda} - s}
= \max_{\p \ll \q} \brak{\EP[-X] - \frac{1}{\lambda} \EQ\edg{\frac{d\p}{d\q} 
\log \frac{d\p}{d\q}}}, \quad X \in B_Z;
$$
see e.g. \cite{CL09}. The resulting robust acceptance set is 
$$
A = \bigcap_{\q \in {\cal Q}} \crl{X \in B_Z : \Ent^{\q}(X) \le 0} + B^+.
$$
If the set of discounted trading gains $G$ is as in \eqref{GAV}, one obtains the following 
variant of Proposition \ref{prop:AVaR}:

\begin{proposition} \label{prop:Ent}
If ${\cal Q}$ is convex, $\sigma({\cal P}_Z, C_Z)$-closed and satisfies 
\be \label{phicondexp}
\sup_{\q \in {\cal Q}} \EQ \exp(\varphi(Z)) < + \infty \mbox{ for 
an increasing function }\varphi : [1,+\infty) \to \mathbb{R} \mbox{ with }\lim_{x \to +\infty} 
\frac{\varphi(x)}{x} = + \infty,
\ee
the following are equivalent:
\begin{itemize}
\item[{\rm (i)}] 
there exist no $X \in G-A$ and $\varepsilon \in \mathbb{R}_{++}$ such that $X \ge \varepsilon$
\item[{\rm (ii)}] there exists no $X \in G-A$ such that $X(\omega) > 0$ for all $\omega \in \Omega$
\item[{\rm (iii)}]
there exists a $\p \in {\cal P}_Z$ such that $\EP X \le 0$ for all $X \in U_Z \cap (G-A)$
\item[{\rm (iv)}] 
$\phi$ is real-valued on $B_Z$ with $\phi(0) = 0$ and 
$\phi(X) = \max_{\p \in {\cal P}_Z} (\mathbb{E}^{\p} X - \eta(\p))$ for all $X \in U_Z$,
\end{itemize}
where $\eta : {\cal P}_Z \to \mathbb{R}_+ \cup \crl{+ \infty}$ is given by
$$
\eta(\p) = \left\{ \begin{array}{ll}
\inf_{\q \in {\cal Q}, \, \p \ll \q} \EQ \brak{\frac{d\p}{d\q} \log \frac{d\p}{d\q}}/\lambda & \mbox{ if $\p$ 
satisfies {\rm a)--b)} and } \p \ll \q \mbox{ for some } \q \in {\cal Q}\\
+ \infty & \mbox{ otherwise},
\end{array} \right.
$$
and {\rm a)--b)} are the same conditions as in Proposition \ref{prop:AVaR}.
\end{proposition}

\begin{Examples} \label{ExQexp} For $Z = 1 + \sum_{j,t \ge 1} \tilde{S}^j_t$, the following are convex
$\sigma({\cal P}_Z, C_Z)$-closed subsets of ${\cal P}_Z$ satisfying \eqref{phicondexp}:\\
{\bf 1.} 
All $\q \in {\cal P}_Z$ such that $c^j_t \le \EQ (\tilde{S}^j_t)^2 \le C^j_t$ 
and $\EQ \exp(\varepsilon^j_t (\tilde{S}^j_t)^2) \le D^j_t$ for given constants $0 \le c^j_t \le C^j_t$ and
$\varepsilon^j_t, D^j_t > 0$.\\
{\bf 2.} All $\q \in {\cal P}_Z$ such that $c^j_t \le \EQ [(\tilde{S}^j_t/\tilde{S}^j_{t-1}-1)^2 \mid {\cal F}_{t-1}] \le C^j_t$ 
and $\EQ \exp(\varepsilon^j_t (\tilde{S}^j_t)^2) \le D^j_t$ for given constants $0 \le c^j_t \le C^j_t$ and
$\varepsilon^j_t, D^j_t > 0$.\\
{\bf 3.} The $\sigma({\cal P}_Z, C_Z)$-closed convex hull of any set ${\cal Q} \subseteq {\cal P}_Z$
satisfying \eqref{phicondexp}.\\[2mm]
Note that Example \ref{ExQ}.3 also satisfies condition \eqref{phicondexp}. Therefore, its $\sigma({\cal P}_Z,C_Z)$-closed
convex hull fulfills the assumptions of Proposition \ref{prop:Ent}.
\end{Examples}

\appendix
\setcounter{equation}{0}
\section{Proofs}
\label{sec:proofs}

\subsection{Representation of increasing convex functionals}

In preparation for the proof of Theorem \ref{thm:main} we first derive representation results for
general increasing convex functionals on $C_Z$ and $U_Z$. As in Section \ref{sec:main}, $\Omega$ 
is a non-empty subset of $(\mathbb{R}_{++} \times \mathbb{R}^{J})^T$ and $Z : \Omega \to [1,+\infty)$ a continuous function 
such that $\crl{\omega \in \Omega : Z(\omega) \le z}$ is compact for all $z \in \mathbb{R}_+$. If $(X_n)$ is 
a sequence of functions $X_n : \Omega \to \mathbb{R}$ 
decreasing pointwise to a function $X : \Omega \to \mathbb{R}$, we write
$X_n \downarrow X$. The space $C_Z$ of continuous functions $X : \Omega \to \mathbb{R}$
such that $X/Z$ is bounded forms a Stone vector lattice; that is, it is a linear space 
with the property that for all $X,Y \in C_Z$, the point-wise minima $X \wedge Y$ 
and $X \wedge 1$ also belong to $C_Z$. Let $ca^+_Z$ be the set of all Borel measures 
$\mu$ satisfying $\ang{Z, \mu} := \int Z d\mu < + \infty$. We call a functional $\psi: C_{Z} \to \mathbb{R}$
increasing if $\psi(X) \ge \psi(Y)$ for $X \ge Y$ and define
the convex conjugate $\psi^*_{C_{Z}} : ca^+_{Z} \to \mathbb{R} \cup \crl{+ \infty}$ by 
$$
\psi^*_{C_{Z}}(\mu) := \sup_{X \in C_{Z}} (\ang{X, \mu} - \psi(X)).
$$

\begin{theorem} \label{thm:rep}
Let $\psi: C_Z \to \mathbb{R}$ be an increasing convex functional with the property that
for every $X \in C_{Z}$ there exists a constant $\varepsilon > 0$ such that
\begin{equation} \label{condepsilon}
\lim_{z \to +\infty} \psi(X + \varepsilon (Z - z)^+) = \psi(X).
\end{equation}
Then
\be \label{repC}
\psi(X)= \max_{\mu \in ca^+_{Z}} (\ang{X,\mu} -\psi^\ast_{C_{Z}}(\mu))
\qquad\mbox{for all } X \in C_{Z}.
\ee
\end{theorem}

\begin{proof}
Fix $X \in C_{Z}$. It is immediate from the definition of $\psi^*_{C_{Z}}$ that 
\be \label{Young}
\psi(X) \ge \sup_{\mu \in ca^+_{Z}} (\ang{X,\mu} - \psi^*_{C_{Z}}(\mu)).
\ee
On the other hand, the Hahn--Banach extension theorem (for example, in the form of Theorem 5.53 in \cite{Ali})
applied to the trivial subspace $\crl{0} \subseteq C_Z$ and the dominating function
$\psi_X : C_Z \to \mathbb{R}$, given by $\psi_X(Y) := \psi(X+Y) - \psi(X)$, yields the existence of a
linear functional $\zeta_X : C_Z \to \mathbb{R}$ that is dominated by $\psi_X$. Since $\psi_X$
is increasing, the same must be true for $\zeta_X$. So if we can show that for every 
sequence $(X_n)$ in $C_{Z}$ satisfying $X_n \downarrow 0$, there exists 
a constant $\eta > 0$ such that $\psi_X(\eta X_n) \downarrow 0$, 
then $\zeta_X(X_n) \downarrow 0$, and we obtain from the Daniell--Stone theorem 
(see e.g., Theorem 4.5.2 in \cite{D}) that $\zeta_X$ is of the form $\zeta_X(Y) = \ang{Y,\mu_X}$ for a 
measure $\mu_X \in ca^+_{Z}$. It follows that $\ang{X,\mu_X} - \psi(X) \ge \ang{X+Y,\mu_X} - \psi(X+Y)$
for all $Y \in C_Z$. In particular, $\psi^*_{C_{Z}}(\mu_X) = \ang{X,\mu_X} - \psi(X)$,
and the representation \eqref{repC} follows from \eqref{Young}.

Now, choose an $\varepsilon > 0$ such that \eqref{condepsilon} holds and $m > 0$
so that $X_1 \le m Z$. Set $\eta = \varepsilon/(4m)$ and fix $\delta > 0$. 
There exists a $z \in \mathbb{R}_+$ such that $\psi_X(\varepsilon (Z-z)^+) \le \delta$, and by our
assumptions on $Z$, the set $\Lambda = \crl{Z \le 2 z}$ is compact. Therefore, one obtains from Dini's lemma that 
$$
x_n := \max_{\omega \in \Lambda} X_n(\omega) \downarrow 0.
$$
Since $x \mapsto \psi_X(x)$ is a convex function from $\mathbb{R}$ to $\mathbb{R}$,
it is continuous. In particular, there exists an $n_0$ such that $\psi_X(2 \eta x_n) \le \delta$
for all $n \ge n_0$. Moreover, it follows from
$$
X_n \le X_n 1_{\crl{Z \le 2 z}} + X_1 1_{\crl{Z > 2 z}}
\le x_n 1_{\crl{Z \le 2 z}} + m Z 1_{\crl{Z > 2 z}} \le x_n + 2m(Z - z)^+
$$
that
$$
\frac{X_n - x_n}{2m} \le (Z - z)^+,
$$
and therefore, 
$$
\psi_X \brak{2 \eta (X_n - x_n)} = \psi_X \brak{\varepsilon \frac{X_n - x_n}{2m}} \le \delta
\quad \mbox{for all } n.
$$
This gives
$$
\psi_X(\eta X_n) \le \frac{\psi_X(2 \eta x_n) + \psi_X(2 \eta (X_n - x_n))}{2} \le \delta
\quad \mbox{for all } n \ge n_0.
$$
Hence, $\psi_X(\eta X_n) \downarrow 0$, and the proof is complete.
\end{proof}

To extend the representation \eqref{repC} beyond $C_Z$, we need the following version 
of condition \eqref{condepsilon}:
\be \label{everycall}
\lim_{z \to +\infty} \psi(n(Z - z)^+) = \psi(0) \quad \mbox{for every } n \in \mathbb{N}.
\ee
The subsequent lemma shows that it implies \eqref{condepsilon}.

\begin{lemma} \label{lemma:cond}
An increasing convex functional $\psi : C_Z \to \mathbb{R}$ with the property \eqref{everycall}
also satisfies \eqref{condepsilon}.
\end{lemma}

\begin{proof}
If \eqref{everycall} holds, one has for any $X \in C_{Z}$,
$\varepsilon \in \mathbb{R}_+$ and $\lambda \in (0,1)$,
$$
\psi(X + \varepsilon (Z - z)^+) \le
\lambda \psi \brak{\frac{X}{\lambda}} + (1-\lambda) \psi \brak{\varepsilon \frac{(Z - z)^+}{1-\lambda}}
$$
and
$$
\psi \brak{\varepsilon \frac{(Z - z)^+}{1-\lambda}} \to \psi(0)
\quad \mbox{for } z \to + \infty.
$$
Moreover, since $z \mapsto \psi(z X)$ is a real-valued convex 
function on $\mathbb{R}$, it is continuous. In particular, 
$$
\lambda \psi \brak{\frac{X}{\lambda}} \to \psi(X) \quad \mbox{and} \quad
(1-\lambda) \psi(0) \to 0 \quad \mbox{for } \lambda \to 1.
$$
This shows that $\psi(X + \varepsilon(Z-z)^+) \to \psi(X)$ for $z \to + \infty$.
\end{proof}
Before giving conditions under which a representation of the form \eqref{repC} holds on $U_Z$,
we note that for a given $X \in U_Z$,
$$ \label{approxXn}
X_n(\omega) := \sup_{\omega' \in \Omega} \brak{\frac{X(\omega')}{Z(\omega')} - n \sum_{j,t}
|\omega^j_t - \omega'^j_t|} Z(\omega)
$$
defines a sequence in $C_{Z}$ such that $X_n \downarrow X$. In addition to this fact, 
we need the following two auxiliary results:

\begin{lemma} \label{lemma:comp}
For every increasing convex functional $\psi : C_{Z} \to \mathbb{R}$ the following hold:

\begin{itemize}
\item[{\rm (i)}]
There exists an increasing convex function 
$\varphi : \mathbb{R}_+ \to \mathbb{R} \cup \crl{+\infty}$ satisfying 
$\lim_{x \to + \infty} \varphi(x)/x = +\infty$ such that $\psi^*_{C_{Z}}(\mu) \ge \varphi(\ang{Z,\mu})$ 
for all $\mu \in ca^+_{Z}$.
\item[{\rm (ii)}]
If $\psi$ satisfies \eqref{everycall}, the sublevel sets $\{\mu \in ca^+_{Z} : \psi^*_{C_{Z}}(\mu) \le a\}$,
$a \in \mathbb{R}$, are $\sigma(ca^+_{Z},C_{Z})$-compact.
\end{itemize}
\end{lemma}

\begin{proof}
For every $\mu \in ca^+_{Z}$, one has
$$
\psi^*_{C_{Z}}(\mu) \ge \sup_{y \in \mathbb{R}_+} (\ang{yZ,\mu} - \psi(yZ)) = \varphi(\ang{Z,\mu})
$$
for the increasing convex function $\varphi : \mathbb{R}_+ \to \mathbb{R} \cup \crl{+\infty}$ given by 
$$
\varphi(x) := \sup_{y \in \mathbb{R}_+} (xy - \psi(y Z)).
$$ 
It follows from the fact that $\psi$ is real-valued that $\lim_{x \to +\infty} \varphi(x)/x = + \infty$.
This shows (i).

As the supremum of $\sigma(ca^+_{Z},C_{Z})$-continuous functions, $\psi^*_{C_Z}$ is 
$\sigma(ca^+_{Z},C_{Z})$-lower semicontinuous. Therefore, the sets
$\Lambda_a := \{\mu \in ca^+_{Z} : \psi^*_{C_{Z}}(\mu) \le a\}$ are $\sigma(ca^+_{Z},C_{Z})$-closed.
Moreover, every $\mu \in \Lambda_a$ satisfies
$$
m \ang{(Z- z)^+,\mu} - \psi(m(Z- z)^+) \le \psi^*_{C_{Z}}(\mu) \le a  \quad \mbox{for all } 
m, z \in \mathbb{R}_+.
$$
So if \eqref{everycall} holds, there exists for every $m \in \mathbb{R}_+$ a $z \in \mathbb{R}_+$ such that 
$$
\ang{(Z- z)^+,\mu} \le \frac{a + \psi(0) + 1}{m} \quad \mbox{for all } \mu \in \Lambda_a.
$$
In particular,
$$
\lim_{z \to +\infty} \sup_{\mu \in \Lambda_a} \ang{(Z- z)^+,\mu} = 0,
$$
and, as a result,
\be \label{tight1}
\lim_{z \to +\infty} \sup_{\mu \in \Lambda_a}  \ang{Z 1_{\crl{Z > 2z}}, \mu} 
\le \lim_{z \to +\infty} \sup_{\mu \in \Lambda_a} \ang{2(Z- z)^+,\mu} = 0.
\ee
From (i) we know that 
\be \label{tight2}
\ang{Z,\mu} \le \varphi^{-1}(a) < + \infty \quad \mbox{for all }\mu \in \Lambda_a,
\ee
where $\varphi^{-1} : \mathbb{R} \to \mathbb{R}_+$ is the right-continuous inverse of $\varphi$, given by 
$$
\varphi^{-1}(y) := \sup \crl{x \in \mathbb{R}_+ : \varphi(x) \le y} \quad \mbox{with } \sup \emptyset := 0.
$$
The mapping $f : X \mapsto X/Z$ identifies $C_{Z}$ with the space of bounded 
continuous functions $C_b$, and $g : \mu \mapsto Z d\mu$ identifies $ca^+_{Z}$ with the set 
of all finite Borel measures $ca^+$. It follows from \eqref{tight1} and \eqref{tight2} that $g(\Lambda_a)$ is tight.
So one obtains from the direct half of Prokhorov's theorem (see e.g., Theorem 8.6.7 in \cite{Bog})
that $g(\Lambda_a)$ is $\sigma(ca^+, C_b)$-compact,
which is equivalent to $\Lambda_a$ being $\sigma(ca^+_{Z}, C_{Z})$-compact.
\end{proof}

\begin{lemma} \label{lemma:conv}
Let $\alpha : ca^+_Z \to \mathbb{R} \cup \crl{+\infty}$ be a mapping such that
$\inf_{\mu \in ca^+_Z} \alpha(\mu) \in \mathbb{R}$ and $\alpha(\mu) \ge \varphi(\ang{Z,\mu})$ 
for all $\mu \in ca^+_{Z}$ and a function $\varphi : \mathbb{R}_+ \to \mathbb{R} \cup \crl{+\infty}$ satisfying
$\lim_{x \to + \infty} \varphi(x)/x = +\infty$. Then
the following hold:
\begin{itemize}
\item[{\rm (i)}]
$\psi(X) := \sup_{\mu \in ca^+_Z} (\ang{X,\mu} - \alpha(\mu))$
defines an increasing convex functional $\psi : B_Z \to \mathbb{R}$.
\item[{\rm (ii)}]
If all sublevel sets $\crl{\mu \in ca^+_Z : \alpha(\mu) \le a}$, $a \in \mathbb{R}$, are relatively 
$\sigma(ca^+_Z, C_Z)$-compact, then $\psi(X_n) \downarrow \psi(X)$
for every sequence $(X_n)$ in $C_Z$ such that $X_n \downarrow X$ for an $X \in C_Z$.
\item[{\rm (iii)}]
If all sublevel sets $\crl{\mu \in ca^+_Z : \alpha(\mu) \le a}$, $a \in \mathbb{R}$, are 
$\sigma(ca^+_Z, C_Z)$-compact, then $\psi(X_n) \downarrow \psi(X)$
for every sequence $(X_n)$ in $C_Z$ such that $X_n \downarrow X$ for an $X \in U_Z$, and
$$
\psi(X) = \max_{\mu \in ca^+_Z} \brak{\ang{X,\mu} - \alpha(\mu)} \quad \mbox{for all } X \in U_Z.
$$
\end{itemize}
\end{lemma}

\begin{proof}
It is clear that $\psi$ is increasing and convex. Moreover, for every $X \in B_Z$, there exists an
$m \in \mathbb{R}_+$ such that $|X| \le mZ$. Therefore,
$$
\psi(X) \ge \sup_{\mu \in ca^+_Z} \brak{-m \ang{Z,\mu} - \alpha(\mu)} > - \infty
$$
as well as
$$
\psi(X) \le \sup_{\mu \in ca^+_Z} \brak{m \ang{Z,\mu} - \alpha(\mu)}
\le \sup_{\mu \in ca^+_Z} \brak{m \ang{Z,\mu} - \varphi(\ang{Z,\mu})} < + \infty.
$$
This shows (i).

Now, let $(X_n)$ be a sequence in $C_Z$ such that $X_n \downarrow X$ for some $X \in U_Z$.
By replacing $\varphi$ with 
$$
\tilde{\varphi}(x) = \inf_{y \ge x} \varphi(y) \vee \inf_{\mu \in ca^+_Z} \alpha(\mu),
$$
one can assume that $\varphi$ is increasing. Then the right-continuous inverse 
$\varphi^{-1} : \mathbb{R} \to \mathbb{R}_+$, given by
$$
\varphi^{-1}(y) := \sup \crl{x \in \mathbb{R}_+ : \varphi(x) \le y} \quad \mbox{with } \sup \emptyset := 0,
$$
satisfies $\lim_{y \to +\infty} \varphi^{-1}(y)/y = 0$ and $\ang{Z,\mu} \le \varphi^{-1}(\alpha(\mu))$ 
for all $$
\mu \in {\rm dom}\, \alpha := \crl{\nu \in ca^+_{Z} : \alpha(\nu) < + \infty}.
$$
Choose $m \in \mathbb{R}_+$ such that $X_1 \le m Z$. Then
\be \label{alarge}
\ang{X_n,\mu} - \alpha(\mu) \le m \ang{Z, \mu} - \alpha(\mu)
\le m \varphi^{-1}(\alpha(\mu)) - \alpha(\mu) \quad
\mbox{for all } n \mbox{ and } \mu \in {\rm dom} \, \alpha.
\ee
If the sets $\crl{\mu \in ca^+_{Z} : \alpha(\mu) \le a}$, $a \in \mathbb{R}$, are 
$\sigma(ca^+_Z, C_Z)$-compact, then $\alpha$ is $\sigma(ca^+_Z, C_Z)$-lower semicontinuous.
Furthermore, it follows from \eqref{alarge} that
for $a \in \mathbb{R}$ large enough, there exists a sequence $(\mu_n)$ in 
$\crl{\mu \in ca^+_{Z} : \alpha(\mu) \le a}$ such that
$$
\psi(X_n) \le \ang{X_n, \mu_n} - \alpha(\mu_n) + \frac{1}{n} \quad \mbox{for all } n.
$$
As shown in the proof of Lemma \ref{lemma:comp}, the pair $(ca^+_Z,C_Z)$ can be identified with $(ca^+,C_b)$, and 
$\sigma(ca^+,C_b)$ is generated by the Kantorovich--Rubinstein norm (see e.g. Theorem 8.3.2 in \cite{Bog}).
So $\sigma(ca^+_Z, C_Z)$ is metrizable. Therefore, after possibly passing to a subsequence, one can 
assume that $(\mu_n)$ converges to a measure $\mu \in \crl{\mu \in ca^+_{Z} : \alpha(\mu) \le a}$ in $\sigma(ca^+_{Z}, C_{Z})$. 
Then 
$$
\alpha(\mu) \le \liminf_n \alpha(\mu_n).
$$
Moreover, for every $\varepsilon > 0$, there is an $n'$ such that 
$\ang{X_{n'}, \mu} \le \ang{X,\mu} + \varepsilon$.
Now choose $n \ge n'$ such that $\ang{X_{n'}, \mu_n} \le \ang{X_{n'}, \mu} + \varepsilon$.
Then 
$$
\ang{X_n, \mu_n} \le \ang{X_{n'},\mu_n} \le \ang{X_{n'}, \mu} + \varepsilon \le \ang{X,\mu} + 2 \varepsilon,
$$
showing that, $\limsup_n \ang{X_n,\mu_n} \le \ang{X,\mu}$, and therefore,
$$
\inf_n \psi(X_n) \le \limsup_n \brak{\ang{X_n,\mu_n} - \alpha(\mu_n)} \le \ang{X,\mu} - \alpha(\mu) \le \psi(X).
$$
In particular, 
$$\psi(X_n) \downarrow \psi(X) = \max_{\mu \in ca^+_Z} \brak{\ang{X,\mu} - \alpha(\mu)},
$$
which shows (iii).

It remains to prove (ii). To do that we note that if $\alpha$ satisfies the assumption of (ii), the $\sigma(ca^+_Z, C_Z)$-lower
semicontinuous hull $\alpha_*$ has $\sigma(ca^+_Z, C_Z)$-compact sublevel sets and
$$\alpha_*(\mu) \ge \varphi_*(\ang{Z,\mu}) \vee \inf_{\mu \in ca^+_Z} \alpha(\mu) \quad \mbox{for all }
\mu \in ca^+_Z,$$ where $\varphi_*$ is the lower semicontinuous hull of $\varphi$. Since 
$\lim_{x \to + \infty} \varphi_*(x)/x = + \infty$ and
$$
\psi(X) = \sup_{\mu \in ca^+_Z} \brak{\ang{X,\mu} - \alpha_*(\mu)} \quad \mbox{for all } X \in C_Z,
$$
it follows from (iii) that $\psi(X_n) \downarrow \psi(X)$ for every sequence 
$(X_n)$ in $C_Z$ such that $X_n \downarrow X$ for an $X \in C_Z$.
This shows (ii).
\end{proof}

Now we are ready to give our representation result for increasing convex functionals on $U_Z$.
For $\mu \in ca^+_Z$, we define
$$
\psi^*_{U_Z}(\mu) := \sup_{X \in U_Z} (\ang{X,\mu} - \phi(X)).
$$

\begin{theorem} \label{thm:repusc}
Let $\psi : U_{Z} \to \mathbb{R}$ be an increasing convex functional satisfying condition
\eqref{everycall}. Then the following are equivalent:

\begin{itemize}
\item[{\rm (i)}] 
$\psi(X) = \max_{\mu \in ca^+_{Z}} (\ang{X,\mu} - \psi^*_{C_{Z}}(\mu))$ for all $X \in U_{Z}$
\item[{\rm (ii)}] 
$\psi(X_n) \downarrow \psi(X)$ for all $X \in U_{Z}$ and every sequence $(X_n)$ in 
$C_{Z}$ such that $X_n \downarrow X$
\item[{\rm (iii)}]
$\psi(X) = \inf_{Y \in C_{Z}, \, Y \ge X} \psi(Y)$ for all $X \in U_{Z}$
\item[{\rm (iv)}] $\psi^*_{C_{Z}}(\mu) = \psi^*_{U_{Z}}(\mu)$ for all $\mu \in ca^+_{Z}$.
\end{itemize}
\end{theorem}

\begin{proof}
Since, by Lemma \ref{lemma:cond}, \eqref{everycall} implies \eqref{condepsilon}, we obtain from 
Theorem \ref{thm:rep} that 
$$
\psi(X) = \max_{\mu \in ca^+_Z} \brak{\ang{X,\mu} - \psi^*_{C_Z}(\mu)} \quad \mbox{for all } X \in C_Z.
$$
Moreover, by Lemma \ref{lemma:comp}, the sublevel sets
$\{\mu \in ca^+_{Z} : \psi^*_{C_{Z}}(\mu) \le a\}$, $a \in \mathbb{R}$, are $\sigma(ca^+_{Z}, C_{Z})$-compact, 
and there exists a function $\varphi : \mathbb{R}_+ \to \mathbb{R} \cup \crl{+\infty}$ such that
$\lim_{x \to +\infty} \varphi(x)/x = +\infty$ and $\psi^*_{C_{Z}}(\mu) \ge \varphi(\ang{Z,\mu})$
for all $\mu \in ca^+_{Z}$. 

(i) $\Rightarrow $ (ii) is now a consequence of part (iii) of Lemma \ref{lemma:conv}, and 
(ii) $\Rightarrow $ (iii) follows since for every $X \in U_Z$ there exists 
a sequence $(X_n)$ in $C_Z$ such that $X_n \downarrow X$.

(iii) $\Rightarrow$ (iv): One obviously has $\psi^*_{U_{Z}} \ge \psi^*_{C_{Z}}$. On the other hand, 
if for every $X \in U_{Z}$, there is a sequence $(X_n)$ in $C_{Z}$ such 
that $X_n \ge X$ and $\psi(X_n) \downarrow \psi(X)$, then
$$
\sup_n (\ang{X_n, \mu} - \psi(X_n)) \ge \ang{X,\mu} - \psi(X),
$$
from which one obtains $\psi^*_{C_{Z}} \ge \psi^*_{U_{Z}} $.

(iv) $\Rightarrow$ (i):
For given $X \in U_Z$, one obtains from the definition of $\psi^*_{U_{Z}}$ that
$$
\psi(X) \ge \sup_{\mu \in ca^+_{Z}} (\ang{X,\mu} - \psi^*_{U_{Z}}(\mu)) = 
\sup_{\mu \in ca^+_{Z}} (\ang{X,\mu} - \psi^*_{C_{Z}}(\mu)).
$$
Conversely, since there exists a sequence $(X_n)$ in $C_Z$ such that 
$X_n \downarrow X$, we can conclude by (iii) of Lemma \ref{lemma:conv} that
$$
\psi(X) \le \inf_n \psi(X_n) = \max_{\mu \in ca^+_{Z}} (\ang{X,\mu} - \psi^*_{C_{Z}}(\mu)).
$$
\end{proof}

\subsection{Proof of Theorem \ref{thm:main}}

Having Theorems \ref{thm:rep} and \ref{thm:repusc} in hand, we now can prove Theorem \ref{thm:main}.
The implications (iii) $\Rightarrow$ (ii) $\Rightarrow$ (i) as well as (vi) 
$\Rightarrow$ (v) $\Rightarrow$ (iv) $\Rightarrow$ (i) hold without assumption \eqref{condcall}. Indeed, 
(iv) $\Rightarrow$ (i) is obvious.

(iii) $\Rightarrow$ (ii): It follows from (iii) that
$$
\min_{\p \in {\cal P}_Z} \phi^*(\p) = - \max_{\p \in {\cal P}_Z} - \phi^*(\p) = - \phi(0) = 0.
$$
In particular, since $\phi(X) \le 0$ for all $X \in G-A$, there exists a $\p \in {\cal P}_Z$ such that
$$
0 = \phi^*(\p) = \sup_{X \in C_Z}  (\EP X - \phi(X)) \ge \sup_{X \in C_Z \cap (G-A)} (\EP X - \phi(X)) 
\ge \sup_{X \in C_Z \cap (G-A)} \EP X,
$$
and therefore, $\EP X \le 0$ for all $X \in C_Z \cap (G-A)$.

(ii) $\Rightarrow$ (i): Assume there exists an $X \in G-A$ such that $X \ge \varepsilon$ for 
some $\varepsilon \in \mathbb{R}_{++}$. By assumption, $A + B^+ \subseteq A$. So it follows that
$\varepsilon$ belongs to $C_Z \cap (G-A)$. But since $\EP \varepsilon = \varepsilon > 0$ for all 
$\p \in {\cal P}_Z$, this contradicts (ii).

(vi) $\Rightarrow$ (v): If (vi) holds, then $\phi^*(\p) \ge \sup_{X \in U_Z} (\EP X - \phi(X))$ for all 
$\p \in {\cal P}_Z$. So (v) follows from (vi) like (ii) from (iii).

(v) $\Rightarrow$ (iv): Assume there exists an $X \in G-A$ such that $X(\omega)>0$ for all $\omega \in \Omega$ and 
fix a $\p \in {\cal P}_Z$. Since $\p$ is a Borel measure on a measurable subset of $\mathbb{R}^{(J+1)T}$, it is regular.
So there exist an $\varepsilon >0$ and a closed set $F \subseteq \crl{X \ge \varepsilon}$ such that 
$\p[F] > 0$, or equivalently, $\EP 1_F > 0$. But since $\varepsilon 1_F$ belongs to 
$U_Z \cap (G-A)$, this contradicts (v).

Now, we assume \eqref{condcall} and show (i) $\Rightarrow$ (iii).
Since $A + B^+ \subseteq A$ and $A-G$ is convex, the mapping $\phi : B_Z \to \mathbb{R} \cup \crl{\pm \infty}$ 
is increasing and convex. Moreover, it follows from (i) that $\phi(m) = m$ for all $m \in \mathbb{R}$,
implying that $\phi(X) \in \mathbb{R}$ for every bounded $X \in B$. 
By condition \eqref{condcall}, there exists for every $n \in \mathbb{N}$ a
$z \in \mathbb{R}_+$ such that $\phi(n(Z-z)^+) \le 1/n$, and therefore,
$$
\phi \brak{\frac{n}{2}Z} \le \frac{\phi(nz) + \phi(n(Z-z)^+)}{2} \le \frac{nz}{2} + \frac{1}{2n}.
$$
Now, one obtains from monotonicity and convexity that $\phi$ is real-valued on $B_Z$. In addition, it follows 
from \eqref{condcall} that $\phi$ satisfies \eqref{everycall}, and so by
Lemma \ref{lemma:cond}, also \eqref{condepsilon}. Therefore, Theorem \ref{thm:rep}
yields $\phi(X) = \max_{\mu \in ca^+_Z} (\ang{X,\mu} - \phi^*_{C_Z}(\mu))$ for all $X \in C_Z$.
Since $\phi(m) = m$ for all $m \in \mathbb{R}$, $\phi^*_{C_Z}(\mu)$ is $+ \infty$ for all $\mu \in ca^+_Z \setminus {\cal P}_Z$,
and one obtains $\phi(X) = \max_{\p \in {\cal P}_Z} (\EP X - \phi^*(\p))$ 
for all $X \in C_Z$.

Finally, if conditions \eqref{condcall}--\eqref{condusc} and (iii) of Theorem \ref{thm:main} hold,
$\phi$ is a real-valued increasing convex functional on $B_Z$ fulfilling \eqref{everycall} as well as
condition (iii) of Theorem \ref{thm:repusc}. So Theorem \ref{thm:repusc} yields 
that $\phi$ also satisfies condition (i) of Theorem \ref{thm:repusc}, and therefore,
(vi) of Theorem \ref{thm:main}, which completes the proof.
\qed

\subsection{Proofs of Propositions \ref{prop:A} and \ref{prop:dualcond}}

{\bf Proof of Proposition \ref{prop:A}}\\
If the acceptance set $A$ is of the form \eqref{acc} for a mapping $\alpha : {\cal P}_Z \to \mathbb{R}_+ \cup \crl{+\infty}$
satisfying (A1)--(A2), it can be written as
$$
A = \crl{X \in B_Z : \rho(X) \le 0} + B^+ \quad \mbox{for} \quad \rho(X) := \sup_{\p \in {\cal P}_Z} \brak{\EP(-X) - \alpha(\p)}.
$$
By passing to the lower convex hull, one can assume that $\beta$ is convex. 
Then, Jensen's inequality yields
$$
\alpha(\p) \ge \EP \beta(Z) \ge \beta(\EP Z) \quad \mbox{for all } \p \in {\cal P}_Z.
$$
By identifying $(C_Z, ca^+_Z)$ with $(C_b, ca^+)$ like in the proof of Lemma \ref{lemma:comp}, one
deduces from Prokhorov's theorem that the sets $\{\p \in {\cal P}_Z : \EP \beta(Z) \le a\}$, $a \in \mathbb{R}$, are 
$\sigma({\cal P}_Z, C_Z)$-compact. As a consequence, the sets $\{\p \in {\cal P}_Z : \alpha(\p) \le a\}$, $a \in \mathbb{R}$, are
relatively $\sigma({\cal P}_Z, C_Z)$-compact, and one obtains from part (ii) of Lemma \ref{lemma:conv} that for every
$n \in \mathbb{N}$,
$$
\rho \brak{\frac{1}{n} - n(Z-z)^+} = - \frac{1}{n} + \rho \brak{- n(Z-z)^+} \downarrow - \frac{1}{n}
\quad \mbox{for } z \to + \infty.
$$
In particular, $1/n - n(Z-z)^+ \in A \subseteq A-G$ for $z$ large enough, 
showing that condition \eqref{condcall} holds. \qed

\bigskip \noindent
{\bf Proof of Proposition \ref{prop:dualcond}}\\
By our assumptions on $G$ and $A$, one has $\phi(0) \le 0$. Moreover, if $\phi(0) = - \infty$, $G-A$
contains $\mathbb{R}$, which by \eqref{condcall}, implies $G-A = B_Z$. Then 
$\phi \equiv - \infty$ on $B_Z$, $\phi^* \equiv + \infty$, and all the statements of Proposition \ref{prop:dualcond} 
become obvious. On the other hand, if $\phi(0) > - \infty$, it follows from \eqref{condcall}, like in the proof 
of Theorem \ref{thm:main}, that $\phi$ is real-valued on $B_Z$. Then
$$
\phi^*(\p) \le \sup_{X \in U_Z} (\EP X - \phi(X)) \quad \mbox{for all } \p \in {\cal P}_Z,
$$
and since it follows from \eqref{condcall} that $\phi$ satisfies \eqref{everycall},
Theorem \ref{thm:repusc} yields that the inequality is an equality if and only if 
$\phi$ satisfies condition \eqref{condusc}. Next, note that since
$$
X - \phi(X) - \varepsilon \in C_Z \cap (G-A) \quad \mbox{for all } X \in C_Z \mbox{ and } \varepsilon > 0,
$$
one has
$$
\phi^*(\p) = \sup_{X \in C_Z} \EP(X-\phi(X)) = \sup_{X \in C_Z \cap (G-A)} \EP X,
$$
and analogously,
$$
\sup_{X \in U_Z} (\EP X-\phi(X)) = \sup_{X \in U_Z \cap (G-A)} \EP X.
$$
This completes the proof of Proposition \ref{prop:dualcond}.
\qed

\subsection{Proofs of Propositions \ref{prop:trans} and \ref{prop:short}}

{\bf Proof of Proposition \ref{prop:trans}}\\
We start by showing the implications (iv) $\Rightarrow$ (iii) $\Rightarrow$ (ii) $\Rightarrow$ (i).
If (iv) holds, one has $0 = \phi(0) = \max_{\p \in {\cal P}_Z} - \phi^*_G(\p)$, yielding (iii).
Moreover, since $\EP X \le 0$ for every $X \in G$ and $\p \in {\cal M}$, (iii) implies (ii).
That (ii) implies (i) is obvious. 

To prove (i) $\Rightarrow$ (iv), we first note that \eqref{sscondition} implies 
\eqref{condcall} and (i) is a reformulation of condition (i) 
of Theorem \ref{thm:main} in the case $A = B^+$. So, by Theorem \ref{thm:main}, it follows from
(i) that $\phi$ is real-valued on $B_Z$ with $\phi(0) = 0$. Moreover, we know from Proposition \ref{prop:dualcond} that
$$
\phi^*(\p) \le \sup_{X \in U_Z \cap (G-B^+)} \EP
\le \sup_{X \in G} \EP X = \phi^*_G(\p), \quad \p \in {\cal P}_Z.
$$
So if we can show that
\be \label{G*}
\phi^*(\p) \ge \phi_G^*(\p) \quad \mbox{for all } \p \in {\cal P}_Z,
\ee
we obtain $\phi^* = \phi^*_G$, and by Proposition \ref{prop:dualcond}, condition \eqref{condusc} holds.
Then it follows from Theorem \ref{thm:main} that (i) implies (iv). To prove \eqref{G*}, we observe that
\[
\sup_{X \in G} \EP X 
= \sup_{\vartheta} \EP \Big[\sum_{t,j \ge 1} \vartheta^j_t \Delta \tilde{S}^j_t -
\frac{g^j_{t-1}(\Delta\vartheta^j_{t} S^j_{t-1})}{S^0_{t-1}}\Big]
+ \sup_\theta \Big( \sum_i \theta_i \EP H_i- h(\theta)\Big),
\]
and since $g^j_{t-1}(0) = 0$, the first supremum can be taken over strategies $\vartheta$ such that 
$$
\EP \Big[\sum_{t,j \ge 1} \vartheta^j_t \Delta \tilde{S}^j_t - \frac{g^j_{t-1}(\Delta\vartheta^j_{t} S^j_{t-1})}{S^0_{t-1}}\Big] \ge 0.
$$
Therefore, $$\EP \edg{\sum_{t,j \ge 1} \vartheta^j_t \Delta \tilde{S}^j_t -g^j_{t-1}(\Delta \vartheta^j_{t} S^j_{t-1})/S^0_{t-1}}$$ can 
be approximated by $$\EP \edg{\sum_{t,j \ge 1} \tilde{\vartheta}^j_t \Delta \tilde{S}^j_t -g^j_{t-1}
(\Delta \tilde{\vartheta}^j_{t} S^j_{t-1})/S^0_{t-1}}$$
for continuous ${\cal F}_{t-1}$-measurable functions $\tilde \vartheta^j_t : \Omega \to \mathbb{R}$
with compact support. But since
$$
\sum_{t,j \ge 1} \brak{\tilde{\vartheta}^j_t \Delta S_t^j- \frac{g^j_{t-1}(\Delta \tilde{\vartheta}^j_t S^j_{t-1})}{S^0_{t-1}}}
+ \sum_i \theta_i H_i- h(\theta) 
$$
is in $C_Z \cap G$, it follows that $\phi^*(\p) \ge \sup_{X \in C_Z \cap G} \EP X \ge \phi^*_G(\p)$ 
for all $\p \in {\cal P}_Z$.

It remains to show that $\phi^*_G$ is of the form \eqref{penaltyG}. To do this we note that
$$
\sum_{t=1}^T \vartheta^j_t\Delta \tilde{S}^j_t=\sum_{t=1}^T \sum_{s=1}^t \Delta \vartheta^j_s\Delta \tilde{S}^j_t=
\sum_{s=1}^T \sum_{t=s}^T \Delta \vartheta^j_s \Delta \tilde{S}^j_t=\sum_{t=1}^T\Delta\vartheta^j_t (\tilde{S}^j_T-\tilde{S}_{t-1}^j).
$$
Hence,
\[\sup_{X \in G} \EP X 
= \sup_{\vartheta} \EP \edg{\sum_{t,j \ge 1}\Delta\vartheta^j_t  (\tilde{S}^j_T- \tilde{S}_{t-1}^j)- \frac{g^j_{t-1}
(\Delta\vartheta^j_{t} S^j_{t-1})}{S^0_{t-1}}}
+ \sup_\theta \Big( \sum_i \theta_i \EP H_i- h(\theta)\Big),\]
where the first supremum can be taken over strategies $\vartheta$ such that 
$$
\EP \Big[\sum_{t,j \ge 1}\Delta\vartheta^j_t  (\tilde{S}^j_T- \tilde{S}_{t-1}^j)-
\frac{g^j_{t-1}(\Delta\vartheta^j_{t} S^j_{t-1})}{S^0_{t-1}} \Big] \ge 0.
$$
Now $\EP \edg{\sum_{t,j \ge 1}\Delta\vartheta^j_t  (\tilde{S}^j_T- \tilde{S}_{t-1}^j)-
g^j_{t-1}(\Delta\vartheta^j_{t} S^j_{t-1})/S^0_{t-1}}$ can be approximated by 
\beas
&& \EP \edg{\sum_{t,j \ge 1} \Delta \tilde{\vartheta}^j_t  (\tilde{S}^j_T- \tilde{S}_{t-1}^j) -
\frac{g^j_{t-1} (\Delta\tilde{\vartheta}^j_{t} S^j_{t-1})}{S^0_{t-1}}}\\ 
&=& \EP \edg{\sum_{t,j \ge 1} \Delta \tilde{\vartheta}^j_t (\EP [\tilde{S}^j_T \mid {\cal F}_{t-1}] - \tilde{S}_{t-1}^j)
- \frac{g^j_{t-1}(\Delta \tilde{\vartheta}^j_{t} S^j_{t-1})}{S^0_{t-1}}}
\eeas
for bounded ${\cal F}_{t-1}$-measurable mappings $\Delta \tilde{\vartheta}^j_t$ with compact support.
On $\{S^j_{t-1} > 0\}$, one has
\beas
&& \sup_{\Delta \tilde{\vartheta}^j_t} \brak{\Delta \tilde{\vartheta}^j_t (\EP[\tilde{S}^j_T\mid{\cal F}_{t-1}]-\tilde{S}_{t-1}^j)
- \frac{g^j_{t-1}(\Delta \tilde{\vartheta}^j_{t} S^j_{t-1})}{S^0_{t-1}}}\\
&=& \frac{1}{S^0_{t-1}} \sup_{\Delta \tilde{\vartheta}^j_t} \Bigg(\Delta \tilde{\vartheta}^j_t S^j_{t-1}
\frac{\EP[\tilde{S}^j_T \mid{\cal F}_{t-1}] - \tilde{S}_{t-1}^j}{\tilde{S}^j_{t-1}}- g^j_{t-1}(\Delta \tilde{\vartheta}^j_{t} S^j_{t-1})\Bigg)\\
&=& \frac{1}{S^0_{t-1}} g^{j*}_{t-1}\Bigg(\frac{\EP[\tilde{S}^j_T\mid{\cal F}_{t-1}]-\tilde{S}_{t-1}^j}{\tilde{S}^j_{t-1}}\Bigg),
\eeas
and on $\{S^j_{t-1} = 0\}$, 
\beas
&& \sup_{\Delta \tilde{\vartheta}^j_t} \brak{\Delta \tilde{\vartheta}^j_t (\EP[\tilde{S}^j_T\mid{\cal F}_{t-1}]- \tilde{S}_{t-1}^j)
- \frac{g^j_{t-1}(\Delta \tilde{\vartheta}^j_{t} S^j_{t-1})}{S^0_{t-1}}}\\
&=& \sup_{\Delta \tilde{\vartheta}^j_t} \Delta \tilde{\vartheta}^j_t \EP[\tilde{S}^j_T\mid{\cal F}_{t-1}]
=+\infty 1_{\{\EP[\tilde{S}^j_T\mid{\cal F}_{t-1}] > 0\}}.
\eeas
Since $\p[S^j_{t-1} = 0 \mbox{ and } \EP[\tilde{S}^j_T > 0 \mid {\cal F}_{t-1}]>0] > 0$ if and only if 
$\p[S^j_{t-1} = 0 \mbox{ and } S^j_T > 0] > 0$, this proves \eqref{penaltyG}.
\qed

\bigskip
\noindent
{\bf Proof of Proposition \ref{prop:short}}\\
It follows as in the proof of Proposition \ref{prop:trans} that 
$$
\phi^*(\p) = \sup_{X \in U_Z} \brak{\EP X - \phi(X)} = \phi^*_G(\p) \quad \mbox{for all } \p \in {\cal P}_Z,
$$
and (i)--(iv) are equivalent. Moreover, 
$$
\phi^*_G(\p) =\sup_{\vartheta \ge 0} \EP\edg{\sum_{t,j \ge 1}\vartheta_t^j \Delta \tilde{S}_t^j} +\sup_{\theta \in \Theta} 
\EP \brak{\sum_{i \in I} \theta_i H_i - h_i},
$$
and since the first supremum can be taken over non-negative predictable strategies $(\vartheta_t)$ 
such that $$\EP \edg{\sum_{t,j \ge 1}\vartheta_t^j \Delta \tilde{S}_t^j} \ge 0,$$ it can equivalently be taken 
over bounded non-negative predictable strategies. Now, it is easy to see that
$$
\phi^*_G(\p) = \left\{\begin{array}{ll}
\sup_{\theta} \EP(\sum_{i \in I} \theta_i H_i - h(\theta)) & \mbox{ if } \tilde{S}^1, \dots, \tilde{S}^J 
\mbox{ are supermartingales under } \p\\
+ \infty & \mbox{ otherwise.}
\end{array} \right.
$$
\qed

\subsection{Proofs of Lemma \ref{lemma:oce} and Propositions \ref{prop:AVaR} and \ref{prop:Ent}}

{\bf Proof of Lemma \ref{lemma:oce}}\\
It follows from (l1)--(l2) that for all $X \in B_Z$, $s \mapsto \EQ l_{\q}(s-X)$ 
is a real-valued increasing convex function on $\mathbb{R}$ such that
$$
\lim_{s \to \pm \infty} \EQ l_{\q}(s-X) - s = + \infty.
$$
In particular, there exists a minimizing $s$, and it is easy to see that $\rho_{\q}$ is a real-valued 
decreasing convex functional on $B_Z$ with the translation property $\rho_{\q}(X+m) = \rho_{\q}(X) - m$, $m \in \mathbb{R}$.
Moreover, $\EQ l_{\q}(cZ) < +\infty$ for all $c \in \mathbb{R}_+$. So
$B_Z$ is contained in the Orlicz heart $M^{l_{\q}}$ corresponding to $\q$ and the Young function 
$l_{\q}(.)-l_{\q}(0)$. By Theorem 4.6 and the computation in Section 5.4 of \cite{CL09},
$$
\rho_{\q}(X) = \max_{\p} \brak{\EP[-X] - \EQ \edg{l^*_{\q}\brak{\frac{d\p}{d\q}}}}, \quad X \in B_Z,
$$
where the maximum is over all $\p \ll \q$ such that $d\p/d\q$ is in the norm-dual of $M^{l_{\q}}$.
For all these $\p$, one has $\EP Z = \EQ[Z d\p/ d\q] < + \infty$, showing that $\p \in {\cal P}_Z$.
On the other hand, since $l_{\q}(x) \ge xy - l^*_{\q}(y)$ for all $x,y \in \mathbb{R}$,
one has for every $X \in B_Z$, $s \in \mathbb{R}$ and $\p \ll \q$, 
$$
\EQ l_{\q}(s-X) - s \ge \EQ \edg{(s-X) \frac{d\p}{d\q}} - \EQ \edg{l^*_{\q}\brak{\frac{d\p}{d\q}}} -s=
\EP[-X] - \EQ \edg{l^*_{\q}\brak{\frac{d\p}{d\q}}}.
$$
This proves the duality \eqref{OrDual}. By Jensen's inequality and (l3), 
$$
\EQ l^*_{\q}(d\p/d\q) \ge l^*_{\q} \EQ [d\p/d\q] = l^*_{\q}(1) = 0,$$ and therefore,
$\min_{\p \in {\cal P}_Z} \EQ l^*_{\q}(d\p/d\q) = \EQ l^*_{\q}(d\q/d\q) = 0$, showing that $\rho_{\q}(0) = 0$.
Moreover, since
$$
\bigcap_{\q\in{\cal Q}} \crl{X\in B_Z: \rho_\q(X)\le 0} = \{X \in B_Z : \sup_{\q \in {\cal Q}} \rho_{\q}(X) \le 0\},
$$
the acceptance set $A$ can be written as
$$
A = \crl{X \in B_Z : \EP X + \alpha(\p) \ge 0 \mbox{ for all } \p \in {\cal P}_Z} + B^+,
$$
where $\alpha(\p) := \inf_{\q \in {\cal Q}} \alpha_{\q}(\p)$ with
$$
\alpha_\q(\p):=\left\{
\begin{array}{ll}
\EQ l^*_\q (d\p/d\q) & \mbox{ if } \p \ll \q\\
+ \infty & \mbox{ otherwise.}
\end{array}
\right.
$$
It follows from $\min_{\p \in {\cal P}_Z} \alpha_{\q}(\p) = 0$ for all $\q \in {\cal Q}$ that
$\min_{\p \in {\cal P}_Z} \alpha(\p) = 0$. Hence (A1) holds.
Moreover, since $l^*_{\q}(y) \ge xy - l_{\q}(x)$ for all $x,y \in \mathbb{R}$, 
one has for every $\p \in {\cal P}_Z$ for which there exists a $\q \in {\cal Q}$ such that $\p \ll \q$,
\beas
\alpha(\p) = \inf_{\q \in {\cal Q}, \, \p \ll \q} \EQ l^*_{\q}\brak{\frac{d\p}{d\q}} \ge
\inf_{\q \in {\cal Q}, \, \p \ll \q} \EQ \edg{\varphi(Z) \frac{d\p}{d\q} - l_{\q} (\varphi(Z))}
\ge \EP \varphi(Z) - \sup_{\q \in {\cal Q}} \EQ l_{\q} (\varphi(Z)),
\eeas
which implies (A2) for $\beta = \varphi - \sup_{\q \in {\cal Q}} \EQ l_{\q} (\varphi(Z))$.
\qed

\bigskip
\noindent
{\bf Proof of Proposition \ref{prop:AVaR}}\\
Since $\AV^{\q}$ is a transformed loss risk measure with loss function $l_{\q}(x) = x^+/\lambda$, it 
follows from the integrability condition \eqref{phicond} that 
the assumptions of Lemma \ref{lemma:oce} are satisfied. So one obtains from 
Proposition \ref{prop:A} that condition \eqref{condcall} holds. Moreover, $G-A$ is a convex cone. Therefore, 
by Proposition \ref{prop:dualcond},
\be \label{GA}
\phi^*(\p) = \sup_{X \in C_Z \cap(G-A)} \EP X = 
\left\{\begin{array}{cl}
0 & \mbox{ if } \EP X \le 0 \mbox{ for all } X \in C_Z \cap (G-A)\\
+ \infty & \mbox{ otherwise.}
\end{array} \right.
\ee
Let us denote $\hat{\cal M} := \crl{\p \in {\cal P}_Z : \phi^*(\p) = 0}$ and write
${\cal M} = {\cal M}_G \cap {\cal M}_A$, where ${\cal M}_G$ is the set of all 
$\p \in {\cal P}_Z$ satisfying conditions a)--b) of Proposition \ref{prop:AVaR}
and ${\cal M}_A$ the set of all $\p \in {\cal P}_Z$ fulfilling condition c). 
Since for $X \in A$, the negative part $X^-$ belongs to $B_Z$, one obtains from \eqref{GA},
\be \label{GA2}
\phi^*(\p) \le \sup_{X \in U_Z \cap(G-A)} \EP X
\le \sup_{X \in G, \, \EP X > - \infty} \EP X - \inf_{X \in A} \EP X \quad \mbox{for all } \p \in {\cal P}_Z.
\ee
It follows as in the proof of Proposition \ref{prop:trans} that the second supremum 
in \eqref{GA2} is zero for all $\p \in {\cal M}_G$, while it can be seen from the dual representation
$$
\AV^{\q}(X) = \sup_{\p \in {\cal P}_Z, \, d\p/d\q \le 1/\lambda} \EP[-X]
$$
that the infimum is zero for all $\p \in {\cal M}_A$. In particular, $ \hat{\cal M}
\supseteq {\cal M} = {\cal M}_G \cap {\cal M}_A$.

On the other hand, it can be shown as in the proof of Proposition \ref{prop:trans}
that $\hat{\cal M} \subseteq {\cal M}_G$. Moreover, 
it follows from our assumptions on ${\cal Q}$ that ${\cal M}_A$ is convex and $\sigma({\cal P}_Z,C_Z)$-closed.
Indeed, for $\p_1, \p_2 \in {\cal M}_A$ and $0 \le \mu \le 1$, there exist  $\q_1, \q_2 \in {\cal Q}$ together with
$Y_1,Y_2 \in B^+$ bounded by $1/\lambda$ such that $\p_1 = Y_1 \cdot \q_1$ and
$\p_2 = Y_2 \cdot \q_2$. Therefore,
$$
(\mu \p_1 + (1-\mu) \p_2)[E] \le \frac{1}{\lambda} (\mu \q_1 + (1-\mu) \q_2)[E]
$$
for all measurable sets $E$. It follows that
$$
\frac{d(\mu \p_1 + (1-\mu) \p_2)}{d(\mu \q_1 + (1-\mu) \q_2)} \le \frac{1}{\lambda},
$$
showing that ${\cal M}_A$ is convex. Furthermore, if $(\p_n)$ is a sequence in ${\cal M}_A$ 
converging to a $\p \in {\cal P}_Z$ in $\sigma({\cal P}_Z,C_Z)$, there exist $\q_n \in {\cal Q}$
and $Y_n \in B^+$ bounded by $1/\lambda$ such that $\p_n = Y_n \cdot \q_n$. Condition \eqref{phicond} 
implies that $\sup_{\q \in {\cal Q}} \EQ Z < +\infty$ and $\sup_{\q \in {\cal Q}} \EQ[Z 1_{\crl{Z > z}}] \to 0$ for $z \to + \infty$. 
So it follows like in the proof of Lemma \ref{lemma:comp} that ${\cal Q}$ is $\sigma({\cal P}_Z,C_Z)$-compact. 
Hence, by passing to a subsequence, one can assume that $\q_n$ converges to a $\q$ in ${\cal Q}$ 
with respect to $\sigma({\cal P}_Z,C_Z)$. Then 
\be \label{ineqC}
\EP X = \lim_n \mathbb{E}^{\p_n} X \le \frac{1}{\lambda} \lim_n \mathbb{E}^{\q_n} X = \frac{1}{\lambda} \EQ X
\quad \mbox{for all } X \in C^+_Z.
\ee
As probability measures on $\mathbb{R}^{(J+1)T}$, $\p$ and $\q$ are regular. Therefore, it follows from 
\eqref{ineqC} that $d\p/d\q \le 1/\lambda$, showing that ${\cal M}_A$ is $\sigma({\cal P}_Z,C_Z)$-closed.
By a separating hyperplane argument, one obtains for every $\hat{\p} \in {\cal P}_Z \setminus {\cal M}_A$,
an $X \in C_Z$ such that $\mathbb{E}^{\hat{\p}} X < \inf_{\p \in {\cal M}_A} \EP X = 0$,
implying $X \in A$ and $\phi^*(\hat{\p}) = \sup_{X \in C_Z \cap (G-A)} \mathbb{E}^{\hat{\mathbb{P}}} X = +\infty$. So 
$\hat{\cal M} \subseteq {\cal M}_A$, and as a consequence $\hat{\cal M} = {\cal M}$. This shows that
$$
\phi^*(\p) = \sup_{X \in U_Z \cap(G-A)} \EP X = 
\left\{\begin{array}{ll}
0 & \mbox{ if } \p \in {\cal M}\\
+\infty & \mbox{ otherwise,}
\end{array} \right.
$$
and it follows from Proposition \ref{prop:dualcond} that \eqref{condusc} holds. As a result,
all conditions (i)--(vi) of Theorem \ref{thm:main} are equivalent, which implies that the 
conditions (i)--(iv) of Proposition \ref{prop:AVaR} are equivalent.
\qed

\bigskip
\noindent
{\bf Proof of Proposition \ref{prop:Ent}}\\
$\Ent^{\q}$ is a transformed loss risk measure corresponding to the
loss function $l_{\q}(x) = \exp(\lambda x -1)/\lambda$. Therefore, it follows from condition \eqref{phicondexp} 
that Lemma \ref{lemma:oce} applies. So we know that condition \eqref{condcall} holds.
As in the proof of Proposition \ref{prop:AVaR}, one has
\be \label{dualentr}
\phi^*(\p) \le \sup_{X \in U_Z \cap (G-A)} \EP X \le \sup_{X \in G, \, \EP X > - \infty} \EP X - \inf_{X \in A} \EP X
\quad \mbox{for all } \p \in {\cal P}_Z,
\ee
and $\sup_{X \in G, \, \EP X > - \infty} \EP X = 0$ for $\p$ in the set ${\cal M}_G$
of all measures in ${\cal P}_Z$ satisfying conditions a)--b). Furthermore, since 
$$
\Ent^{\q}(X) = \sup_{\p \in {\cal P}_Z} (\EP[-X] - \eta_{\q}(\p)) \quad \mbox{for all } X \in B_Z,
$$
where 
$$
\eta_{\q}(\p) = \left\{ \begin{array}{ll}
\EQ \brak{\frac{d\p}{d\q} \log \frac{d\p}{d\q}}/\lambda & \mbox{ if } \p \ll \q\\
+ \infty & \mbox{ otherwise,}
\end{array} \right.
$$
one obtains
$$
\inf_{\q \in {\cal Q}} \eta_{\q}(\p) \ge \sup_{X \in B_Z} \brak{\EP[-X] - \sup_{\q \in {\cal Q}} 
\Ent^{\q}(X)} \ge \sup_{X \in A} \EP[-X]
\ge \phi^*(\p) \quad \mbox{for } \p \in {\cal M}_G.
$$
It follows from the assumptions that there exists a continuous function $\tilde{\varphi} : [1,+\infty) \to \mathbb{R}$ 
such that 
$$
\lim_{x \to + \infty} \frac{\tilde{\varphi}(x)}{x} = + \infty \quad \mbox{and} \quad 
\lim_{x \to + \infty} \frac{\varphi(x)}{\tilde{\varphi}(x)} = + \infty.
$$
Denote $\tilde{Z} = \exp(\tilde{\varphi}(Z))$. Then, it follows from condition
\eqref{phicondexp} that $\sup_{\q \in {\cal Q}} \EQ \tilde{Z} < +\infty$ and \linebreak
$\sup_{\q \in {\cal Q}} \EQ[\tilde{Z} 1_{\crl{\tilde{Z} > z}}] \to 0$ for $z \to + \infty$.
So one obtains as in the proof of Lemma \ref{lemma:comp} that ${\cal Q}$ is 
relatively compact in the topology $\sigma({\cal P}_{\tilde{Z}},C_{\tilde{Z}})$.
But since ${\cal Q}$ was assumed to be $\sigma({\cal P}_Z,C_Z)$-closed, and 
$\sigma({\cal P}_{\tilde{Z}},C_{\tilde{Z}})$ is stronger than $\sigma({\cal P}_Z,C_Z)$, 
${\cal Q}$ is $\sigma({\cal P}_{\tilde{Z}},C_{\tilde{Z}})$-compact. 
Moreover, for all $X \in C_Z$, $\exp(X)$ belongs to $C_{\tilde{Z}}$ and
$$
\Ent^{\q}(X) = \frac{1}{\lambda} \log \EQ \exp(-\lambda X)
$$
is concave as well as $\sigma({\cal P}_{\tilde{Z}},C_{\tilde{Z}})$-continuous in $\q$. Therefore, one obtains 
from a minimax result, such as e.g. the one of Ky Fan \cite{Fan53}, that for all $\p \in {\cal P}_Z$,
\beas
&& \inf_{\q \in {\cal Q}} \eta_{\q}(\p) = \inf_{\q \in {\cal Q}} \sup_{X \in C_Z} \brak{\EP[-X] - \Ent^{\q}(X)}\\
&=& \sup_{X \in C_Z} \brak{\EP[-X] - \sup_{\q \in {\cal Q}} \Ent^{\q}(X)} 
= \sup_{X \in C_Z \cap A} \EP[-X] \le \phi^*(\p).
\eeas
Since $\phi^*(\p) = + \infty$ for $\p \in {\cal P}_Z \setminus {\cal M}_G$, this shows that
$$
\eta(\p) = \left\{ \begin{array}{ll}
\inf_{\q \in {\cal Q}} \eta_{\q}(\p) & \mbox{ if } \p \in {\cal M}_G\\
+ \infty & \mbox{ otherwise} \end{array} \right\} = \phi^*(\p),
$$
which, by \eqref{dualentr}, implies $\phi^*(\p) = \sup_{X \in U_Z \cap (G-A)} \EP X$.
So Proposition \ref{prop:dualcond} gives us that condition \eqref{condusc} holds. 
Now Proposition \ref{prop:Ent} is a consequence of Theorem \ref{thm:main}.
 \qed


\end{document}